\documentclass[journal,onecolumn,12pt,draftclsnofoot]{IEEEtran}
\usepackage{amsthm,cite}
\usepackage{amssymb, amsmath, color}
\usepackage{nicefrac}
\usepackage{epsfig}

\usepackage{fullpage}
\usepackage{graphicx}
\usepackage{url}
\usepackage{flushend}
\usepackage{subcaption}
\usepackage{nicefrac}
\usepackage{tkz-euclide}
\usetkzobj{all}
\usepackage{tikz}
\usepackage{epstopdf}

\title{\huge{On the Delay-Storage Trade-off in Content Download from Coded Distributed Storage Systems}}
\author{Gauri Joshi,~%~\IEEEmembership{Student Member,~IEEE,} ( Im not an IEEE member yet)
	 Yanpei Liu,~\IEEEmembership{Student Member,~IEEE,}
  	Emina Soljanin,~\IEEEmembership{Senior Member,~IEEE} \thanks{This
    work was presented in part at the 50th Annual Allerton Conference on
	Communication, Control and Computing, Monticello IL, October % Gauri - Allerton month changed from Sept to Oct
    2012.}  \thanks{G.~Joshi is with the Department of Electrical Engineering and Computer Science, Massachusetts Institute of Technology, Cambridge, MA. Y.~Liu is with the Department of Electrical and
    Computer Engineering, University of Wisconsin Madison, Madison, WI. E.~Soljanin is with Bell Labs, Alcatel-Lucent, Murray Hill, NJ.
    (E-mail: gauri@mit.edu, yliu73@wisc.edu, emina@research.bell-labs.com)}}

\newtheorem{defn}{Definition}
\newtheorem{rem}{Remark}
\newtheorem{lem}{Lemma}
\newtheorem{thm}{Theorem}

\newtheorem{corr}{Corollary}
\newcommand{\expec}{\text{E}}
\newcommand{\var}{\text{V}}

\begin{document}
\maketitle

\begin{abstract}
In this paper we study how coding in distributed storage reduces expected download time, in addition to providing reliability against disk failures. The expected download time is reduced because when a content file is encoded to add redundancy and distributed across multiple disks, reading only a subset of the disks is sufficient to reconstruct the content. For the same total storage used, coding exploits the diversity in storage better than simple replication, and hence gives faster download. We use a novel fork-join queuing framework to model multiple users requesting the content simultaneously, and derive bounds on the expected download time. Our system model and results are a novel generalization of the fork-join system that is studied in queueing theory literature. Our results demonstrate the fundamental trade-off between the expected download time and the amount of storage space. This trade-off can be used for design of the amount of redundancy required to meet the delay constraints on content delivery. 

\end{abstract}


\section{Introduction
\label{sec:intro}}

Large-scale cloud storage and distributed file systems such as Amazon Elastic Block Store (EBS) \cite{EBS} and Google File System
(GoogleFS) \cite{GFS} have become the backbone of many applications such as web searching, e-commerce, and cluster computing. In these distributed storage systems, the content files stored on a set of disks may be simultaneously requested by multiple users. The users have two major demands: reliable storage and fast content download. Content download time includes the time taken for a user to compete with the other users for access to the disks, and the time to acquire the data from the disks. Fast content download is important for delay-sensitive applications such as video streaming, VoIP, as well as collaborative tools like Dropbox \cite{Dropbox} and Google Docs \cite{Googledoc}.

The authors in \cite{GFS} point out that in large-scale distributed storage systems, disk failures are the norm and not the exception. To protect the data from disk failures, cloud storage providers today simply replicate content throughout the storage network over multiple disks. In addition to fault tolerance, replication makes the content quickly accessible since multiple users requesting a content can be directed to different replicas. However, replication consumes a large amount of storage space. In data centers that process massive data today, using more storage space implies higher expenditure on electricity, maintenance and repair, as well as the cost of leasing physical space.

Coding, which was originally developed for reliable communication in presence of noise, offers a more efficient way to store data in distributed systems. The main idea behind coding is to add redundancy so that a content, stored on a set of disks, can be reconstructed by reading a subset of these disks. Previous work shows that coding can achieve the same reliability against failures with lower storage space used. It also allows efficient replacement of disks that have to be removed due to failure or maintenance. We show that in addition to reliability and easy repair, coding also gives faster content download because we only have to wait for content download from a subset of the disks. Some preliminary results on the analysis of download time via queueing-theoretic modeling are presented in \cite{gauri_yanpei_emina_allerton}.
%
 
%\textcolor{blue}{Add one more sentence about our work}
%
\subsection{Previous Work}
%\textcolor{red}{Research in coding for distributed storage was galvanized by the results reported in \cite{dimakis07}. Prior to that work, literature on distributed storage recognized that, when compared with replication, coding can offer huge storage savings for the same reliability levels. But it was also argued that the benefits of coding are limited, and are outweighed by certain disadvantages and extra complexity. Namely, to provide reliability in multi-disk storage systems, when some disks fail, it must be possible to restore either the exact lost data or an equivalent reliability with minimal download from the remaining storage. The cost of this repair regeneration was considered much higher in coded than in replication system \cite{liskov}, until \cite{dimakis07} established existence and advantages of new regenerating codes. This work was then quickly followed, and the area is very active today (see e.g., \cite{kumar, kumar2} and references therein).}

%\textcolor{blue}{
Research in coding for distributed storage was galvanized by the results reported in \cite{dimakis07}. Prior to that work, literature on distributed storage recognized that, when compared with replication, coding can offer huge storage savings for the same reliability levels. But it was also argued that the benefits of coding are limited, and are outweighed by certain disadvantages and extra complexity. Namely, to provide reliability in multi-disk storage systems, when some disks fail, it must be possible to restore either the exact lost data or an equivalent reliability with minimal download from the remaining storage. This problem of efficient recovery from disk failures was addressed in some early work \cite{even_odd}. But in general, the cost of repair regeneration was considered much higher in coded than in replication systems \cite{liskov}, until \cite{dimakis07} established existence and advantages of new regenerating codes. This work was then quickly followed, and the area is very active today (see e.g., \cite{kumar, kumar2, zigzag_codes} and references therein).
%}

%A related line of work is concerned with another potential weakness of coding in distributed storage. Namely, if any part of the data changes, the corresponding disks must be updated accordingly. To minimize the cost of such updates, the authors in \cite{prasanthAllerton10} propose a class of randomized codes which have update complexity scaling logarithmically with the size of data but can correct a linearly scaled number of disk failures. Furthermore, the existence of such update efficient codes that also minimize the repair bandwidth for exact data reconstruction was established in \cite{ankitISIT11}.

%A related line of work is analyzing content accessibility to multiple users requesting the data. 

Only recently \cite{ulric_muriel_emina, berk_isit, mds_queue} was it realized that in addition to reliability, coding can guarantee the same level of content accessibility, but with lower storage than replication. In \cite{ulric_muriel_emina}, the scenario that when there are multiple requests, all except one of them are blocked and the accessibility is measured in terms of blocking probability is considered. In \cite{berk_isit}, multiple requests are placed in a queue instead of blocking and the authors propose a scheduling scheme to map requests to servers (or disks) to minimize the waiting time. In \cite{mds_queue}, the authors give a combinatorial proof that flooding requests to all disks, instead of a subset of them gives the fastest download time. This result corroborates the system model we consider in this paper to model the distributed storage system and analyze its download time.

%\textcolor{red}{Using redundancy in coding for delay reduction has also been studied in packet transmission \cite{kabatiansky_krouk_semenov, Maxemchuk93, mia}, and in some other scenarios of content retrieval in \cite{Allerton10e}. Although they share some common spirit, they do not consider storage systems and the impact of redundancy coding in such scenarios.}

%\textcolor{blue}{
Using redundancy in coding for delay reduction has also been studied in the context of packet transmission in \cite{kabatiansky_krouk_semenov, Maxemchuk93, mia}, and for some content retrieval scenarios in \cite{lihao_thesis, Allerton10e}. Although they share some common spirit, they do not consider the effect of queueing of requests in coded distributed storage systems.
%}

%\textcolor{red}{Talk about Google File System \cite{GFS}, Our allerton paper \cite{gauri_yanpei_emina_allerton}}
%
\subsection{Our Contributions}

In this paper we show that coding allows fast content download in addition to reliable storage. Since multiple users can simultaneously request the content, the download time includes the time to wait for access to the disks plus the time to read the data. When the content is coded and distributed on multiple disks, it is sufficient to read it only from a subset of these disks in order to retrieve the content. We take a queuing-theoretic approach to study how coding the content in this way provides diversity in storage, and achieve a significant reduction in the download time. The analysis of download time leads us to an interesting trade-off between download time and storage space, which can be used to design the optimal level of redundancy in a distributed storage system. To the best of our knowledge, we are the first to propose the $(n,k)$ fork-join system and find bounds on its mean response time, a novel generalization of the $(n,n)$ fork-join system studied in queueing theory literature.

We consider that requests entering the system are assigned to multiple disks, where they enter local queues waiting for disks access. Note that this is in contrast to some existing works (e.g.~\cite{berk_isit}, \cite{mds_queue}) where requests wait in a centralized queue when all disks are busy. Our approach of immediate dispatching of requests to local queues is used by most server farms to facilitate fast acknowledgement response to customers \cite{mor_book}. Under this queueing model, we propose the $(n,k)$ fork-join system, where each request is forked to $n$ disks that store the coded content, and it exits the system when any $k$ ($k \leq n$) disks are read. The $(n,n)$ fork-join system in which all $n$ disks have to be read has been extensively studied in queueing theory and operations research related literature \cite{kim_agarwal, nelson_tantawi, varki_merc_chen}. Our analysis of download time can be seen as a generalization to the analysis of the $(n,n)$ fork-join system.

%\subsection{Organization}
The rest of the paper is organized as follows. In Section~\ref{sec:main_idea}, we present some preliminary concepts that are central to the results presented in the paper. In Section~\ref{sec:n_k_fork_join}, we analyze the expected download time of the $(n,k)$ fork-join system and present the fundamental trade-off between expected download time and storage. These results were presented in part in \cite{gauri_yanpei_emina_allerton}. In Section~\ref{sec:extensions}, we relax some simplifying assumptions, and present the delay-storage trade-off by considering some practical issues such as heavy-tailed and correlated service times of the disks. In Section~\ref{sec:m_n_k_fork_join}, we extend the analysis to distributed storage systems with a large number of disks. Such systems can be divided into groups of $n$ disks each, where each group is an independent $(n,k)$ fork-join system. Finally, Section~\ref{sec:conclu} concludes the paper and gives future research directions.

%The theoretical analysis of download time in Section~\ref{sec:n_k_fork_join} and Section~\ref{sec:m_n_k_fork_join} makes some simplifying assumptions about the distribution of the service time taken to download the content from each disk. As a result, in Section~\ref{sec:extensions}, we present the delay-storage trade-off by generalizing of the service time distribution to take into consideration practical issues such as content preparation time and correlation between the service times. 

\section{Preliminary Concepts}
\label{sec:main_idea}
% Main Idea

\subsection{Reducing Delay using Coding}
One natural way to reduce the download time is to replicate the content across $n$ disks. Then if the user issues $n$ download requests, one to each disk, it only needs to wait for the one of the requests to be served. This strategy gives a sharp reduction in download time, but at the cost of $n$ times more storage space and the cost of processing multiple requests.

It is more efficient to use coding instead of replication. Consider that a content $F$ of unit size is divided into $k$ blocks of equal size. It is encoded to $n \geq k$ blocks using an $(n,k)$ maximum distance separable (MDS) code, and the coded blocks are stored on an array of $n$ disks. MDS codes have the property that any $k$ out of the $n$ blocks are sufficient to reconstruct the entire file. MDS codes have been suggested to provide reliability against disk failures. In this paper we show that, in addition to error-correction, we can exploit these codes to reduce the download time of the content.

The encoded blocks are stored on $n$ different disks (one block per disk). Each incoming request is sent to all $n$ disks, and the content can be recovered when any $k$ out of $n$ blocks are successfully downloaded. An illustrative example with $n=3$ disks and $k=2$ is shown in Fig.~\ref{fig:sys_model_example}. The content $F$ is split into equal blocks $a$ and $b$, and stored on $3$ disks as $a$, $b$, and $a \oplus b$, the exclusive-or of blocks $a$ and $b$. Thus each disk stores content of half the size of file $F$. Downloads from any $2$ disks jointly enable reconstruction of $F$.
\begin{figure}[t]
%\label{fig:sys_model_example}
\begin{center}
\begin{tikzpicture} [scale=0.6]
\def\myangle{180};
\def\dr{25pt}
\def\hr{3pt}
\def\tr{5pt}
\coordinate (A) at (-5,2);
%\node[right] at (-5,5.25) (fji) {user requests are forked to $3$ disks, each containing half };
\draw [black!30, - , fill] (A) -- ($(A)+(\dr,0)$)  arc (0:360:\dr) -- cycle;
\draw [white, fill] (A) circle (\tr);
\draw [black!50] (A) circle (\hr);
\node at ($(A)+(\dr,0)$)(I) {};
\node at ($(A)+(2.75*\dr,0)$)(IF) {};
\foreach \n in {0,2,4} {
\coordinate (A) at (0,\n);
\node at ($(A)-(\dr,0)$)(F-\n) {};
\node at ($(A)+(\dr,0)$)(J-\n) {};
\draw [black!30, - , fill] (A) -- ($(A)+(\dr,0)$)  arc (0:360:\dr) -- cycle;
\draw [red!50, - , fill] (A) -- ($(A)+(\dr,0)$)  arc (0:\myangle:\dr) -- cycle;
\draw [white, fill] (A) circle (\tr);
\draw [black!50] (A) circle (\hr);
}
\path (I) edge[->, gray, very thick] (IF);
\path (IF) edge[->, bend right, gray, very thick] (F-0);
\path (IF) edge[->, bend left, gray, very thick] (F-4);
\path (IF) edge[->, gray, very thick] (F-2);
\draw ($(0,0)+(0,1.1*\tr)$) node [above] {{$a+b$}};
\draw ($(0,2)+(0,1.1*\tr)$) node [above] {{$b$}};
\draw ($(0,4)+(0,1.1*\tr)$) node [above] {{$a$}};
\coordinate (A) at (5,2);
\draw [red!50, - , fill] (A) -- ($(A)+(\dr,0)$)  arc (0:360:\dr) -- cycle;
\draw [black] (A) -- ($(A)+(\dr,0)$);
\draw [black] (A) -- ($(A)-(\dr,0)$);
\draw [white, fill] (A) circle (\tr);
\draw [black!50] (A) circle (\hr);
\node[text width=0.15cm, text height=0.5cm] at ($(A)-(2.75*\dr,0)$)(O) {};
\node[text width=0.2cm,text height=0.75cm] at ($(A)-(3*\dr,0)$)(OL) {};
\node at ($(A)-(0.9\dr,0)$)(OJ) {};
\path (OJ) edge[<-*, bend right, red!50,very thick] (OL);
\path (OJ) edge[<-*, bend left, red!50,very thick] (OL);
\path (J-0) edge[-o, bend right, red!50, very thick] (O.south west);
\path (J-4) edge[-o, bend left,very thick,gray,dashed] (O.north west);
\path (J-2) edge[-o, red!50,very thick] (O);
\draw ($(5,2)+(0,1.1*\tr)$) node [above] {{$a+b$}};
\draw ($(5,2)-(0,1.1*\tr)$) node [below] {{$b$}};
\end{tikzpicture}
\end{center}
\caption{ Storage is $50\%$ higher, but response time (per disk \& overall) is reduced. \label{fig:sys_model_example}}
\end{figure}
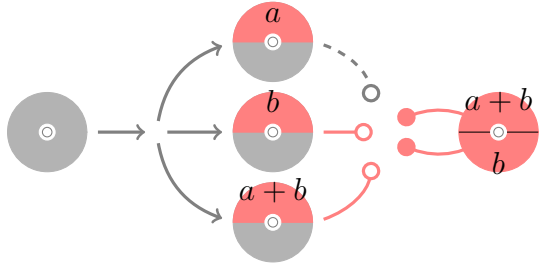

\subsection{Role of Order Statistics}
The time taken to download a block of content is a random variable. If the block download times are independent and identically distributed (i.i.d.), the time to download any $k$ out of $n$ blocks is the $k^{th}$ order statistic of the block download times. We now provide some background on order statistics of i.i.d.\ random variables. For a more complete treatment, please refer to \cite{order_stat}. Although in our system model the block download times are not i.i.d., this background on i.i.d.\ order statistics is a powerful tool for our analysis on the dependent case as shown in later sections.

Let $X_1, \,\, X_2, \,\, \cdots \, ,X_n$ be i.i.d.\ random variables. Then, $X_{k,n}$, the $k^{th}$ order statistic of $X_i$ , $ 1 \leq i \leq n$ , or the $k^{th}$ smallest variable has the distribution,
\begin{align*}
f_{X_{k,n}}(x) &= n \binom{n-1}{k-1} F_{X}(x)^{k-1}(1\!-\! F_{X}(x))^{n-k}f_{X}(x),
\end{align*}
where $f_X$ is the probability density function (PDF) and $F_X$ is the cumulative distribution function (CDF) of $X_i$ for
all $i$. In particular, if $X_i$'s are exponential with mean $1/\mu$, then the expectation and variance of order statistic $X_{k,n}$ are given by,
\begin{align}
\label{eq.orderstat_mean}
\expec[X_{k,n}] &= \frac{1}{\mu}\sum_{i=1}^{k} {\frac{1}{n-k+i}} = \frac{1}{\mu} (H_n - H_{n-k}), \\
\label{eq.orderstat_var}
\var[X_{k,n}] &= \frac{1}{\mu^2} \! \sum_{i=1}^{k} \! {\frac{1}{(n \!-\! k \!+\! i)^2}} \! = \! \frac{1}{\mu^2} (H_{n^2} \!-\! H_{(n-k)^2}),
\end{align}
where $H_n$ and $H_{n^2}$ are generalized harmonic numbers defined by
\begin{align}
\label{eq.harmonic_no}
H_n =  \sum_{j=1}^n \frac{1}{j} ~~ \text{and} ~~ H_{n^2} = \sum_{j=1}^n \frac{1}{j^2}.
\end{align}

We observe from \eqref{eq.orderstat_mean} that for fixed $n$, $\expec[X_{k,n}]$ decreases when $k$ becomes smaller. This fact will help us understand the analysis of download time in Section~\ref{sec:n_k_fork_join} and Section~\ref{sec:m_n_k_fork_join} respectively.

\subsection{Assignment Policies}
\label{subsec:power_of_d_intro}
In our distributed storage model, we divide the content into $k$ blocks, we use $1/k$ units of space of each disk, and hence total storage space used is $n/k$ units. This is unlike conventional replication-based storage solutions where $n$ {\em entire} copies of content are stored on the $n$ disks. In such systems, each incoming request can be assigned to any of the $n$ disks. One such assignment policies is the power-of-$d$ assignment \cite{powerof2, mor_book}. For each incoming request, the power-of-$d$ job assignment uniformly selects $d$ nodes ($d \leq n$) and sends the request to the node with least work left among the $d$ nodes. The amount of work left of a node can be the expected time taken for that node to become empty when there are no new arrivals or simply the number of jobs queued. When $d = n$, power-of-$d$ reduces to the least-work-left (LWL) policy (or joint-the-shortest-queue (JSQ) if work is measured by the number of jobs). Power-of-$d$ assignment has received much attention recently due to the prevailing popularity in large-scale parallel computing. In Section~\ref{sec:n_k_fork_join} and Section~\ref{sec:m_n_k_fork_join}, we compare these policies with our proposed distributed storage model. %In Section~\ref{sec:m_n_k_fork_join} where we consider a storage system with multiple will study how these  assignment policies fit in our model where full copy of content is not available on each disk.

\section{The $(n,k)$ Fork-join System}
\label{sec:n_k_fork_join}

We consider the scenario that users attempt to download the content from the distributed storage system where their requests are placed in a queue at each disk. In Section~\ref{subsec:sys_model_n_k} we propose the $(n,k)$ fork-join system to model the queueing of download requests, and derive theoretical bounds on the expected download time in Section~\ref{subsec:bounds}. This analysis leads us to the fundamental trade-off between download time and storage, which provides insights into the practical system design. Numerical and simulation results demonstrating this trade-off are presented in Section~\ref{subsec:num_results_n_k}.

\subsection{System Model}
\label{subsec:sys_model_n_k}
%Our objective is to determine the download time -- the expected time from the arrival of a request until it finishes service by reading the content from any $k$ of the $n$ disks. We refer to this as the mean response time, the term used in queueing theory literature. of the system --  
We model the queueing of download requests at the disks using the $(n,k)$ fork-join system which is defined as follows.
% {\bf [changed to $n \geq k$ -- Yanpei]}
\begin{defn}[$(n,k)$ fork-join system]
\label{defn:n_k_fork_join}
An $(n,k)$ fork-join system consists of $n$ nodes. Every arriving job is divided into $n$ tasks which enter first-come first-serve queues at each of the $n$ nodes. The job departs the system when any $k$ out of $n$ tasks are served by their respective nodes. The remaining $n-k$ tasks abandon their queues and exit the system before completion of service.% without receiving service. {\bf [They leave maybe partially serviced not without receiving service at all -- Yanpei]}
\end{defn}

The $(n,n)$ fork-join system, known in literature as fork-join queue, has been extensively studied in, e.g., \cite{kim_agarwal, nelson_tantawi, varki_merc_chen}. However, the $(n,k)$ generalization in Definition~\ref{defn:n_k_fork_join} above has not been previously studied to our best knowledge. Fig.~\ref{fig:fork_join_queue} illustrates the $(3,2)$ fork-join system corresponding to the coded distributed storage example shown in Fig.~\ref{fig:sys_model_example}. Each download request, or a job is forked to the $3$ nodes. When $2$ out of $3$ tasks are served, the third task abandons its queue and the job exits the system. For example, Job $1$ is about to exit the system, while Job $2$ is waiting for one more task to be served. The letters $F$ and $J$ denote the fork and join operations respectively.%-- added by Gauri }
%{\bf [Hi Gauri, this paragraph looks good. However when we refer to ``blue" or ``green" jobs, remember that the reviewers are probably reading a printed black/white copy of our paper. This means they don't know what greens/blues refer to. I suggest we use dashed square, solid square and empty square to denote different jobs -- Yanpei]}

\begin{figure}[t]
\centering
\includegraphics[width=3in]{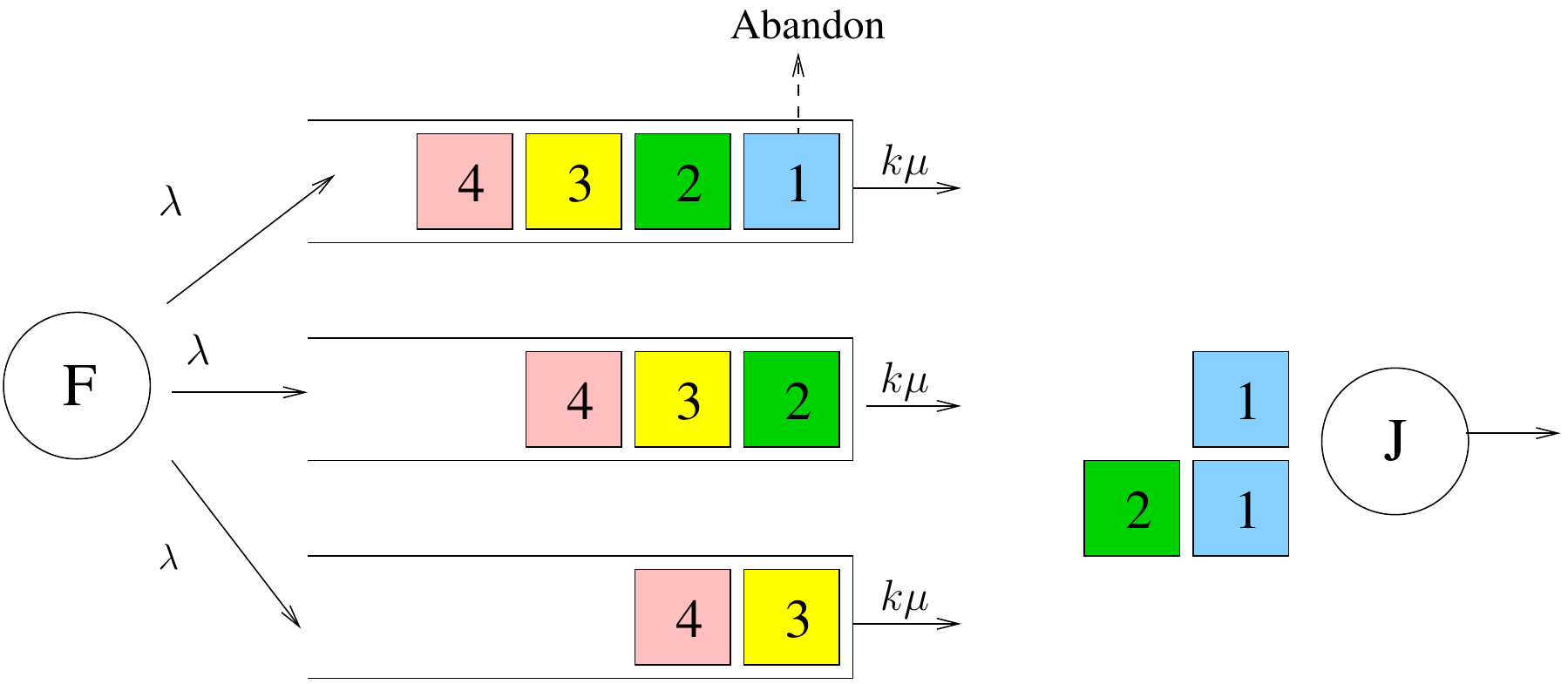}
\caption{Illustration of the $(3,2)$ fork-join system. Since $2$ out of $3$ tasks of Job $1$ are served, the third task abandons its queue and the job exits the system. Job $2$ has to wait for one more task to be served.\label{fig:fork_join_queue}}
\end{figure}

We consider that arrival of download requests is Poisson with rate $\lambda$. Every request is forked to the $n$ disks. The time taken to download one unit of data is exponential with mean $1/\mu$. Since, each disk stores $1/k$ units of data, consider that the service time for each node is exponentially distributed with mean $1/\mu'$ where $\mu' =  k\mu$. Define the load factor $\rho' \triangleq \lambda/\mu'$. This model with an M/M/1 queue at every disk is sometimes referred to as a Flatto–-Hahn-–Wright (or FHW) model \cite{flatto1984two, wright1992two} in fork-join queue literature. While most of our analytical results in Section~\ref{sec:n_k_fork_join} and Section~\ref{sec:m_n_k_fork_join} are for the FHW model, and we use simulations to study systems with M/G/1 queues at the disks in Section~\ref{sec:extensions}.

% case when jobs arrive according to a Poisson process and service times are exponentially distributed is sometimes referred to as a Flatto–-Hahn-–Wright model or FHW model \cite{flatto1984two, wright1992two} in fork-join queue literature where $n=k$. Most of our analytical results will be obtained for the FHW model for $k\le n$. We will also study systems with the general and some specific service times in Sec.~\ref{sec:extensions} by simulation.

For the $(n,n)$ fork-join system to be stable, \cite{fj_stability} shows that the arrival rate $\lambda$ must be less than $\mu'$, the service rate of a node, which in our $(n,n)$ system equals to $n\mu$. In Lemma~\ref{lem:stability_n_k} below, we show that $\lambda<n\mu$ is also a necessary condition for the stability of the $(n,k)$ fork-join system for any $1 \leq k \leq n$.

\begin{lem}[Stability of $(n,k)$ fork-join system]
\label{lem:stability_n_k}
For the $(n,k)$ fork-join system to be stable, the rate of Poisson arrivals $\lambda$ and the service rate $\mu' = k \mu $ per node must satisfy
\begin{align}
\lambda &< \frac{n \mu'}{k} = n\mu. \label{eqn:stability_n_k}
\end{align}
\end{lem}
\begin{proof}
Tasks arrive at each queue at rate $\lambda$ and are served at rate $\mu' = k \mu$. But when $k$ out of the $n$ tasks finish service, the remaining $n-k$ tasks abandon their queues. A task can be one of the abandoning tasks with probability $(n-k)/n$. Hence the effective arrival rate to each queue is $\lambda$ minus the rate of abandonment $\lambda(n-k)/n$. Then the condition for stability of each queue is
\begin{align}
\lambda -\frac{\lambda (n-k)}{n} &< \mu' ,%\\
% \frac{\lambda k}{n} &< k \mu
\end{align}
which reduces to \eqref{eqn:stability_n_k}. %{\bf [I suggest we delete equation (7) and just say right below (6) that ``which reduces to (5) after rearrangement of terms -- Yanpei] - DONE}
\end{proof}

\subsection{Bounds on the Mean Response Time}
\label{subsec:bounds}
% Bound on response time of n,k fork-join system

%---------BEGIN ----  Added by Gauri - 20th July -----------------------------

Our objective is to determine the expected download time, which we refer to as the mean response time $T_{(n,k)}$ of the $(n,k)$ fork-join system. It is the expected time that a job spends in the system, from its arrival until $k$ out of $n$ of its tasks are served by their respective nodes. Previous works \cite{kim_agarwal, nelson_tantawi, varki_merc_chen} have studied $T_{(n,n)}$,  but it has not been found in closed form -- only bounds are known. An exact expression for the mean response time is found only for the $(2,2)$ fork-join system \cite{nelson_tantawi}. %{\bf [edited -- Yanpei] - edited -Gauri}

Since the $n$ tasks are served by independent M/M/1 queues, intuition suggests that $T_{(n,k)}$ is the $k^{th}$ order statistic of $n$ exponential service times. However this is not true, which makes the analysis of $T_{(n,k)}$ challenging. The reason why the order statistics approach does not work is that when $j$ nodes ($j < n$) finish serving their tasks they can start serving the tasks of the next job (cf.~Fig.~\ref{fig:fork_join_queue}). As a result, the service time of a job depends on the departure time of previous jobs.
% {\bf [I would get rid of the word ``naive guess" as we don't want to call reviewers naive if they guess so :P. Also this sentence sounds vague as it is not clear what $n$ exp random variables we are talking about. It could be $n$ M/M/1 mean response time or $n$ exp service time so the readers may wonder which one is more naive than the other -- Yanpei]}
%The analysis of the response time $T_{(n,k)}$ is hard because contrary to intuition, it is not the $k^{th}$ order statistic of $n$ exponential random variables corresponding to the $n$ $M/M/1$ queues.

%An exception is the $(n,1)$ fork-join system in which when one task completes its service all other $n-1$ tasks abandon their respective queues. Hence the tasks of the next job receive their services at the same time and this is true for all the jobs queued behind. As a result the $(n,1)$ fork-join behaves exactly as an M/M/1 queue with arrival rate $\lambda$ and service rate $\mu ' = n \mu$. Thus the response time is exponential with the mean $T_{(n,1)}$ equal to $1/(\mu' - \lambda)$. %{\bf[Please check the 3 paragraphs above carefully -- Gauri] [edited -- Yanpei]}

%The reason why the fork-join system is harder to analyze than a set of parallel independent $M/M/1$ queues is that each incoming job is sent to the $n$ queues. Hence, the arrivals to the queues are perfectly synchronized and the response times of the $n$ queues are correlated.

We now present upper and lower bounds on the mean response time $T_{(n,k)}$. The numerical results in Section~\ref{subsec:num_results_n_k} show that these bounds are fairly tight.

\begin{thm}[Upper Bound on Mean Response Time]
\label{thm:upper_bnd}
The mean response time $T_{(n,k)}$ of an $(n,k)$ fork-join system satisfies
\begin{align}
\label{eqn:upper_bnd}
T_{(n,k)} \leq  & \, \frac{H_n - H_{n-k}}{\mu'}\,  + \\
& \, \frac{\lambda\bigl [(H_{n^2} - H_{ (n-k)^2 }) + (H_n - H_{(n-k)})^2\bigr]}{2 \mu'^2 \bigl[ 1- \rho'(H_n - H_{n-k})\bigr]}, \nonumber
\end{align}
where $\lambda$ is the request arrival rate, $\mu'$ is the service rate at each queue, $\rho' = \lambda/\mu'$ is the load factor,
and the generalized harmonic numbers $H_n$ and $H_{n^2}$ are as given in (\ref{eq.harmonic_no}). The bound is valid only when $\rho' (H_n - H_{n-k}) < 1$.
\end{thm}

\begin{proof}
To find this upper bound, we use a model called the split-merge system, which is similar but easier to analyze than the fork-join system. In the $(n,k)$ fork-join queueing model, after a node serves a task, it can start serving the next task in its queue. On the contrary, in the split-merge model, the $n$ nodes are blocked until $k$ of them finish service. Thus, the job departs all the queues at the same time. Due to this blocking of nodes, the mean response time of the $(n,k)$ split-merge model is an upper bound on (and a pessimistic estimate of) $T_{(n,k)}$ for the $(n,k)$ fork-join system. 

The $(n,k)$ split-merge system is equivalent to an M/G/1 queue where arrivals are Poisson with rate $\lambda$ and service time is a random variable $S$ distributed according to the $k^{th}$ order statistic of the exponential distribution.
% given by
%\[
%F_S(s) = n \binom{ n-1 }{k-1} \mu' e^{-\mu' (n-k+1)s } ( 1- e^{- \mu' s} )^k
%\]
The mean and variance of $S$ are (cf.~(\ref{eq.orderstat_mean}) and (\ref{eq.orderstat_var}))
\begin{equation}
\expec[S] = \frac{H_n - H_{n-k}}{\mu'}  ~~ \text{and} ~~  \var[S] = \frac{H_{n^2} - H_{(n-k)^2}}{\mu'^2}.
\label{eqn:mvk_exp}
\end{equation}
The Pollaczek-Khinchin formula \cite{dsp_gallager} gives the mean response time $T$ of an M/G/1 queue in terms of the mean and variance of $S$ as,
\begin{equation}
T = \expec[S] + \frac{ \lambda ( \expec[S]^2  + \var[S] } {2( 1-\lambda \expec[S])}. \label{eqn:pollac_khin}
\end{equation}
Substituting the values of $\expec[S]$ and $\var[S]$ given by (\ref{eqn:mvk_exp}), we get the upper bound (\ref{eqn:upper_bnd}). Note that the Pollaczek-Khinchin formula is valid only when $ \frac{1}{\lambda} > \expec[S]$, the stability condition of the M/G/1 queue. Since $\expec[S]$ increases with $k$, there exists a $k_0$ such that the M/G/1 queue is unstable for all $k \geq k_0$. The inequality $\frac{1}{\lambda} > \expec[S]$ can be simplified to $\rho' (H_n - H_{n-k}) <1$ which is the condition for validity of the upper bound given in Theorem~\ref{thm:upper_bnd}.
\end{proof}

%\begin{rem}
%For a fixed number of disks $n$, the split-merge system is a good approximation of the fork join system for small $k$ and large $\mu$.
%Hence we expect the upper bound \eqref{eqn:upper_bnd} to become less tight as $k$ increases and/or $\mu$ decreases
%as confirmed by the numerical results in Section~\ref{subsec:num_results_n_k}.
%(cf.~Fig.~\ref{fig:not_tight_bounds}). %{\bf [edited -- Yanpei]}
%\end{rem}

\begin{rem}%[Alternative approach to upper bound]
%{\bf [edited -- Yanpei]}
\label{rem:upper_bnd}
For the $(n,n)$ fork-join system, the authors in \cite{nelson_tantawi} find an upper bound on mean response time different from \eqref{eqn:upper_bnd} derived above. To find the bound, they first prove that the response times of the $n$ queues form a set of associated random variables \cite{assoc_rand_vars}. Then they use the property of associated random variables that their expected maximum is less than that for independent variables with the same marginal distributions. However this approach used in \cite{nelson_tantawi} cannot be extended to the $(n,k)$ fork-join system with $k < n$ because this property of associated variables does not hold for the $k^{th}$ order statistic for $k<n$.
\end{rem}

As a corollary to Theorem~\ref{thm:upper_bnd} above, we can get an exact expression for $T_{(n,1})$, the mean response time of the $(n,1)$ fork-join system. Recall that in the $(n,1)$ fork-join system, the entire content is replicated on $n$ disks, and we just have to wait for any one disk to serve the incoming request.

\begin{corr}
The mean response time $T_{(n,1)}$ of the $(n,1)$ fork-join system is given by
\begin{equation}
T_{(n,1)} = \frac{1}{n\mu - \lambda}, \label{eqn:T_n_1}
\end{equation}
where $\lambda$ is the rate of Poisson arrivals and $\mu$ is the service rate. 
\end{corr}

\begin{proof}
In Theorem~\ref{thm:upper_bnd} we constructed the $(n,k)$ split-merge system which always has worse response time than the corresponding $(n,k)$ fork-join system. For the special case when $k=1$, the split-merge system is equivalent to the fork-join system and gives the same response time. Substituting $k=1$ and $\mu' = k \mu = \mu$ in \eqref{eqn:mvk_exp} and \eqref{eqn:pollac_khin} we get the result \eqref{eqn:T_n_1}.
\end{proof}

\begin{thm}[Lower Bound on Mean Response Time]
\label{thm:lower_bnd}
The mean response time $T_{(n,k)}$ of an $(n,k)$ fork-join queueing system satisfies
\begin{equation}
T_{(n,k)} \!\! \geq \! \frac{1}{\mu'} \! \bigl[ H_n - H_{n-k} + \rho' (H_{n(n-\rho')} - H_{(n-k)(n-k-\rho')}) \bigr], \label{eqn:lower_bnd}
\end{equation}
where $\lambda$ is the request arrival rate, $\mu'$ is the service rate at each queue, $\rho' = \lambda/\mu'$ is the load factor,
and the generalized harmonic number $H_{n(n-\rho')}$ is given by
\begin{equation*}
H_{n(n-\rho')} = \sum_{j=1}^n \frac{1}{j(j-\rho')}.
\end{equation*}
\end{thm}

\begin{proof}
The lower bound in (\ref{eqn:lower_bnd}) is a generalization of the bound for the $(n,n)$ fork-join system derived in \cite{varki_merc_chen}. The bound for the $(n,n)$ system is derived by considering that a job goes through $n$ stages of processing. A job is said to be in the $j^{th}$ stage if $j$ out of $n$ tasks have been served by their respective nodes for $0 \leq j \leq n-1$. The job waits for the remaining $n-j$ tasks to be served, after which it departs the system. For the $(n,k)$ fork-join system, since we only need $k$ tasks to finish service, each job now goes through $k$ stages of processing. In the $j^{th}$ stage, where $0 \leq j \leq k-1$, $j$ tasks have been served and the job will depart when $k-j$ more tasks to finish service.

We now show that the service rate of a job in the $j^{th}$ stage of processing is \emph{at most} $(n-j) \mu'$. Consider two jobs $B_1$ and $B_2$ in the $i^{th}$ and $j^{th}$ stages of processing respectively. Let $i>j$, that is, $B_1$ has completed more tasks than $B_2$. Job $B_2$ moves to the $(j+1)^{th}$ stage when one of its $n-j$ remaining tasks complete. If all these tasks are at the heads of their respective queues, the service rate for job $B_2$ is exactly $(n-j) \mu'$. However since $i > j$, $B_1$'s task could be ahead of $B_2$'s in one of the $n-j$ pending queues, due to which that task of $B_2$ cannot be immediately served. Hence, we have shown that the service rate of in the $j^{th}$ stage of processing is at most $(n-j)\mu'$.

%Consider two jobs $B_1$ and $B_2$ in the $i^{th}$ and $j^{th}$ stages of processing respectively. Let $i>j$, that is, $B_1$ has completed more tasks than $B_2$. Since every incoming job is sent to all $n$ queues, this implies that $B_1$'s tasks will be in front $B_2$'s in all $n-i$ queues that are serving $B_1$'s tasks. Further, we can conclude that the mean service rate of job $B_2$ moving to the $(j+1)^{th}$ stage of processing is at most $(n-j)\mu'$. If the $n-j$ pending tasks are at the head of all the respective queues, then the service rate will be exactly $(n-j)\mu'$. However, $B_1$'s task could be ahead of $B_2$'s in one of the $n-j$ pending queues, due to which that task of $B_2$ cannot be immediately served. Hence, we have shown that for a job in the $j^{th}$ stage of processing, the mean service rate is at most $(n-j)\mu'$.
%
%Thus, the time for a job to move from the $j^{th}$ to $(j+1)^{th}$ stage is lower bounded by the response time of an M/M/1 queue with arrival rate $\lambda$ and service rate $(n-j) \mu'$. The response time is exponentially distributed with mean $T_j = 1/( (n-j)\mu' - \lambda)$. By the memoryless property of the exponential distribution, the total mean response time is the sum of the mean response times of each of the $k$ stages of processing, given by

Thus, the time for a job to move from the $j^{th}$ to $(j+1)^{th}$ stage is lower bounded by $1/( (n-j)\mu' - \lambda)$, the mean response time of an M/M/1 queue with arrival rate $\lambda$ and service rate $(n-j) \mu'$. The total mean response time is the sum of the mean response times of each of the $k$ stages of processing and is bounded below as

\begin{align*}
T_{(n,k)} &\geq \sum_{j=0}^{k-1} \frac{1}{(n-j)\mu' - \lambda},
= \frac{1}{\mu'}  \sum_{j=0}^{k-1} \frac{1}{ (n-j) - \rho'}, \\
&= \frac{1}{\mu'}  \sum_{j=0}^{k-1} \Bigl[ \frac{1}{n-j} + \frac{\rho'}{(n-j)(n-j-\rho')} \Bigr], \\
&= \!\frac{1}{\mu'} \bigl[\! H_n \!-\! H_{n-k} \!+\! \rho' (\!H_{n(n-\rho')}\! -\! H_{(n-k)(n-k-\rho')}\!)\! \bigr].
\end{align*}

\end{proof}

\begin{figure}[!ht]
        \centering
        \begin{subfigure}[hbt]{0.45\textwidth}
                \centering
                 \includegraphics[width=\textwidth]{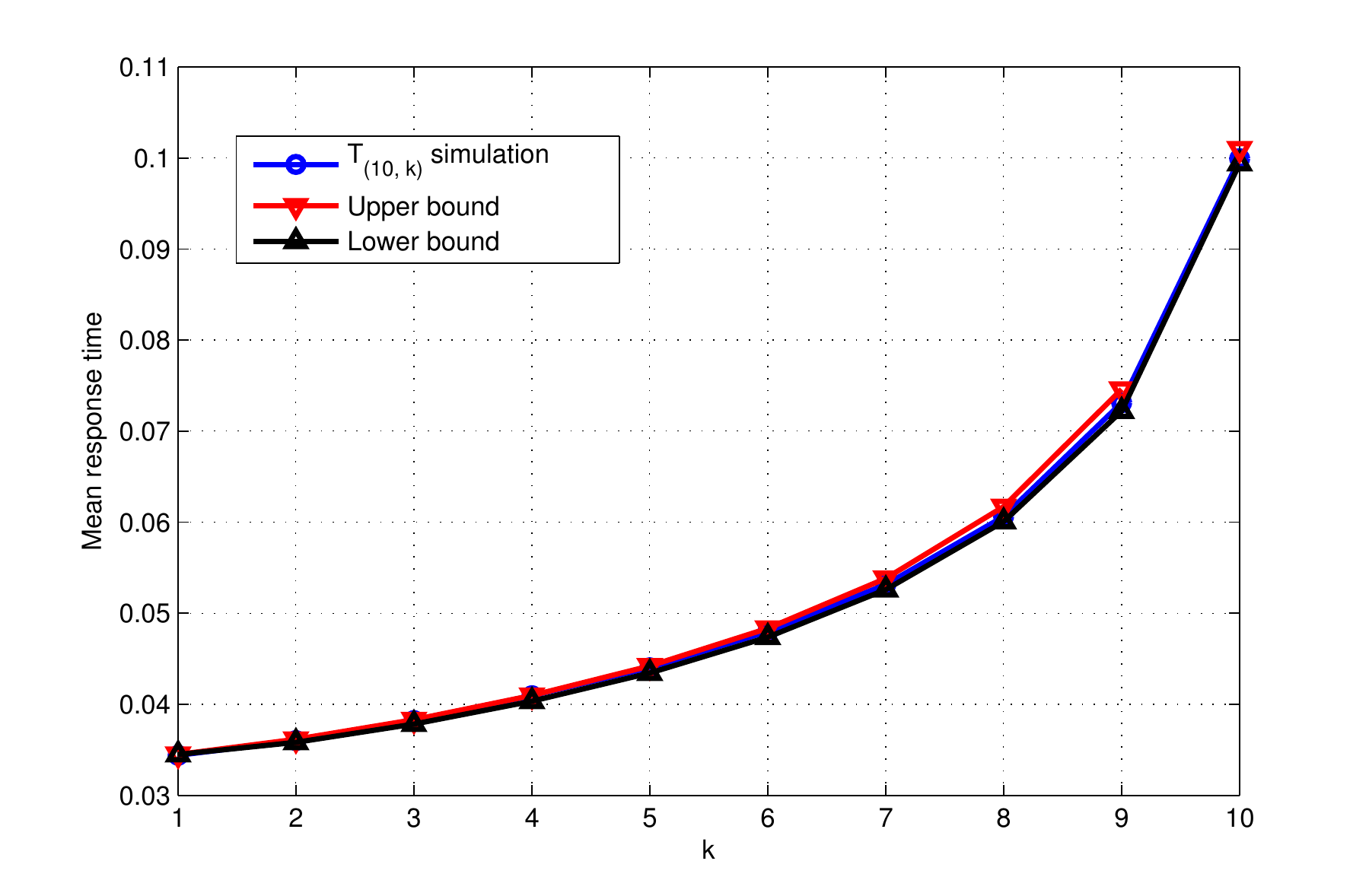}
                \caption{Arrival rate $\lambda = 1$ and service rate $\mu = 3$.}
                \label{fig:nksim}
        \end{subfigure}%
        ~ %add desired spacing between images, e. g. ~, \quad, \qquad etc.

        \begin{subfigure}[hbt]{0.45\textwidth}
                \centering
                \includegraphics[width=\textwidth]{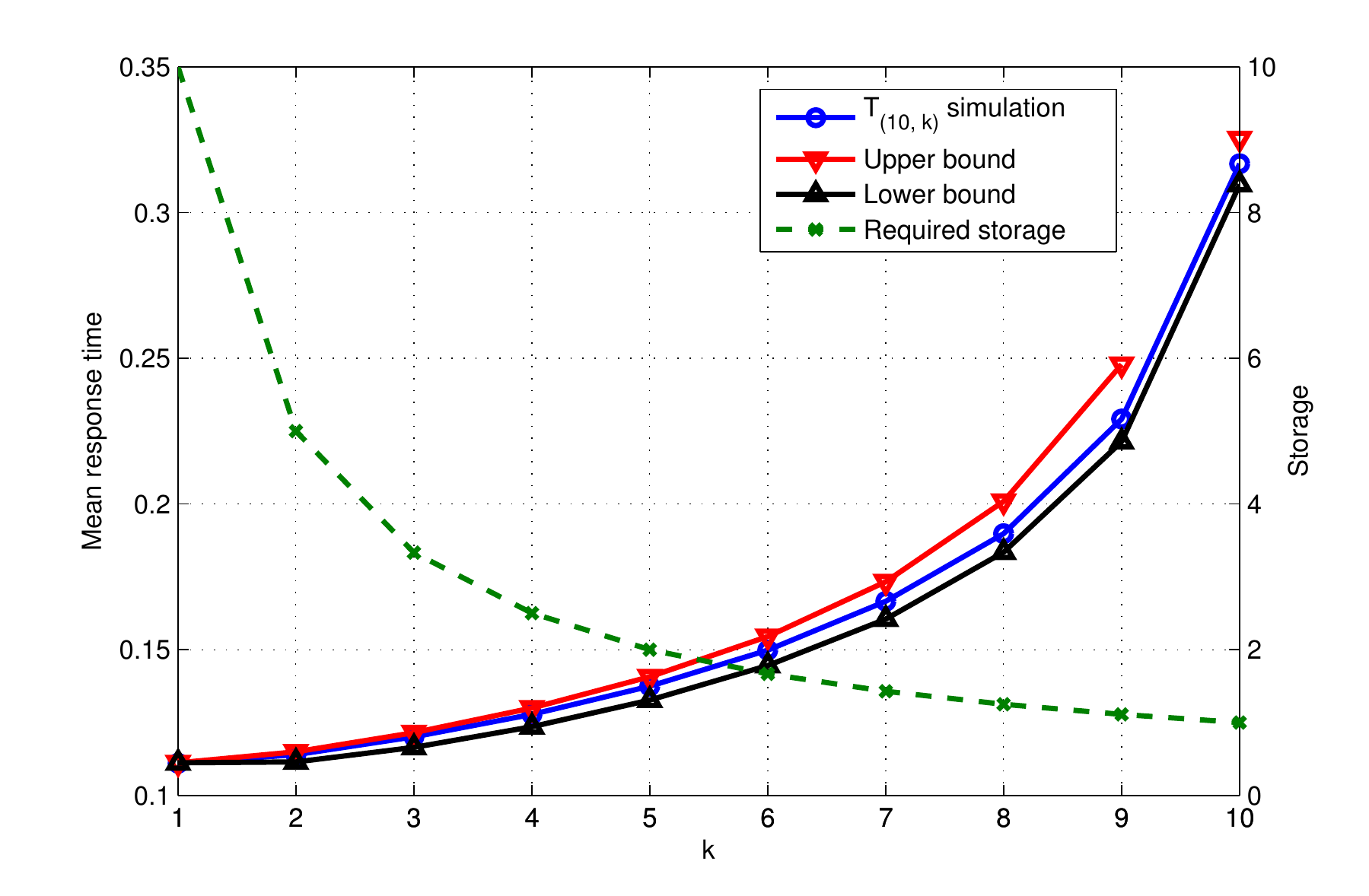}
                \caption{Arrival rate $\lambda = 1$ and service rate $\mu = 1$.}
                \label{fig:not_tight_bounds}
        \end{subfigure}

        \begin{subfigure}[hbt]{0.45\textwidth}
                \centering
                \includegraphics[width=\textwidth]{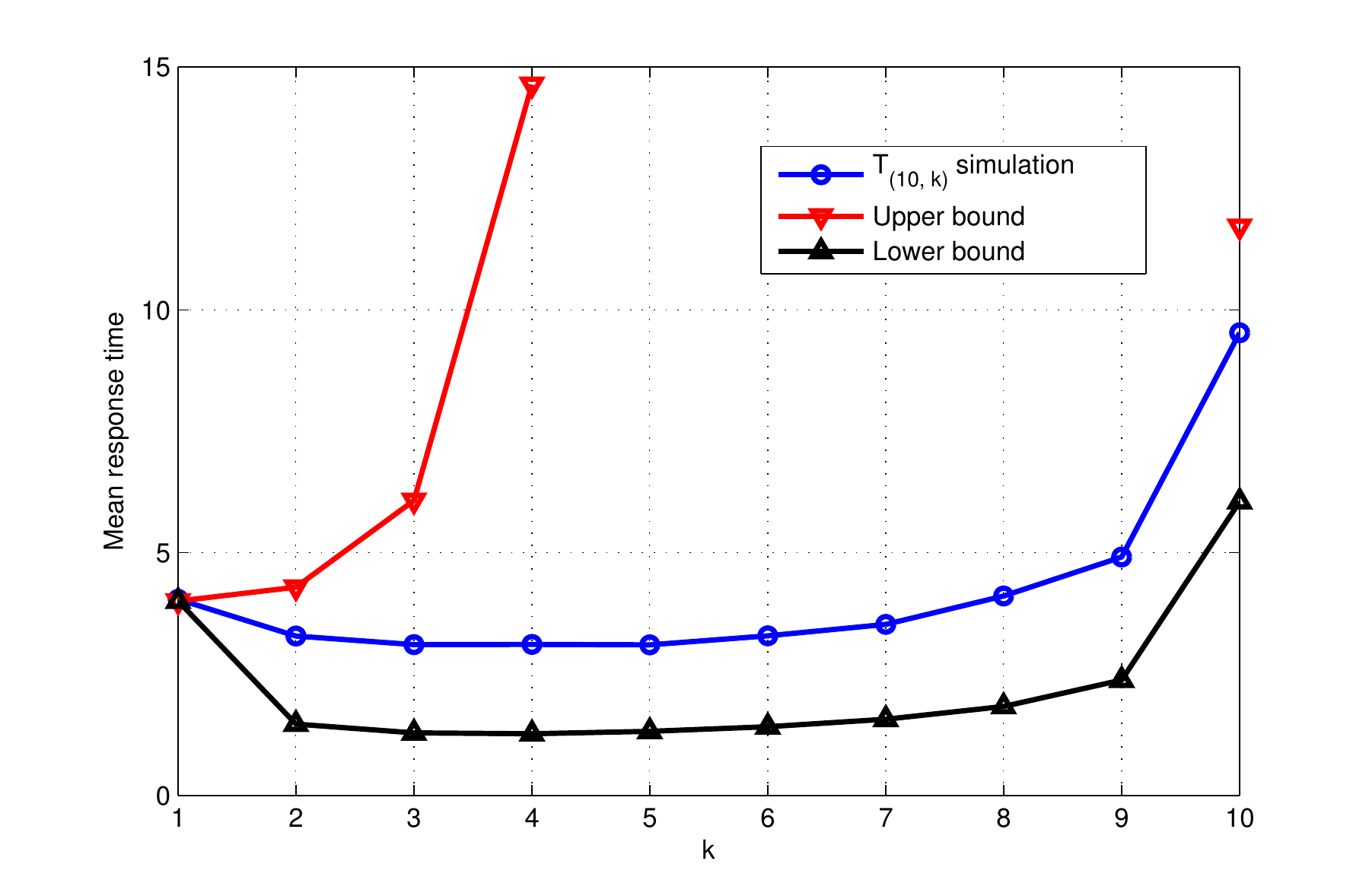}
                \caption{Arrival rate $\lambda = 1$ and service rate $\mu = 1/8$.}
                \label{fig:verybad_bounds}
        \end{subfigure}
        \caption{Behavior of the mean response time $T_{(10, k)}$ as $k$ increases (and total storage $n/k$ decreases). The plot shows that the bounds on mean response time given by \eqref{eqn:upper_bnd} and \eqref{eqn:lower_bnd} are tight when the system is lightly loaded and become loose as $\mu$ decreases and/or $k$ increases.}
        \label{fig.Tnk_sim}
\end{figure}

Hence, we have found lower and upper bounds on the mean response time $T_{(n,k)}$. In Fig.~\ref{fig.Tnk_sim} we demonstrate how the tightness of the bounds changes with service rate $\mu$. The figure shows the mean response time of a $(10, k)$ fork-join system versus $k$ for  service rates $\mu = 3, 1$ and $1/8$. Note that the upper bound for $k=n$ shown in the plot is $T_{(n,n)}\le H_n/(n\mu-\lambda)$ as given in \cite{nelson_tantawi}, instead of the bound in \eqref{eqn:upper_bnd}. The reason behind this substitution is explained in Remark~\ref{rem:upper_bnd}. 

We observe in Fig.~\ref{fig.Tnk_sim} that the bounds become loose as $k$ increases and/or $\mu$ decreases. In particular, the upper bound becomes loose because the blocking of queues in split-merge system becomes significant when $k$ increases and/or $\mu$ decreases. For $\mu = 1/8$, the upper bound in \eqref{eqn:upper_bnd} becomes invalid for $ k \geq 5$ because the condition $\rho' (H_n - H_{n-k}) < 1$ is violated. Similarly, the lower bound becomes loose with increasing $k$ and decreasing $\mu$ because the difference between the actual service rate in the $j^{th}$ stage of processing, and its bound $(n-j) \mu'$ increases. When $k=1$, the bounds coincide and give $T_{(n,1)}=1/(n\mu-\lambda)$.

%Our results show that the bounds are very tight for large $\mu$ when the system is lightly loaded and become loose as $\mu$ decreases. For $\mu = 1/8$, the upper bound in \eqref{eqn:upper_bnd} becomes invalid for $ k \geq 5$ because the condition $\rho' (H_n - H_{n-k}) < 1$ is violated. For a fixed number of disks $n$, the split-merge system is a good approximation of the fork join system for small $k$ and large $\mu$. Hence we expect the upper bound \eqref{eqn:upper_bnd} to become less tight as $k$ increases and/or $\mu$ decreases
%as confirmed by the numerical results in Section~\ref{subsec:num_results_n_k}.
%(cf.~Fig.~\ref{fig:not_tight_bounds}). When $k=1$, the bounds coincide and give $T_{(n,1)}=1/(n\mu-\lambda)$.
%\end{rem}

\subsection{Download Time vs.\ Storage Space Trade-off}
\label{subsec:num_results_n_k}
% Numerical Results
\begin{figure*}[ht]
\begin{center}
\begin{tikzpicture}
\node at (0,0) (P) {\includegraphics[scale=0.55]{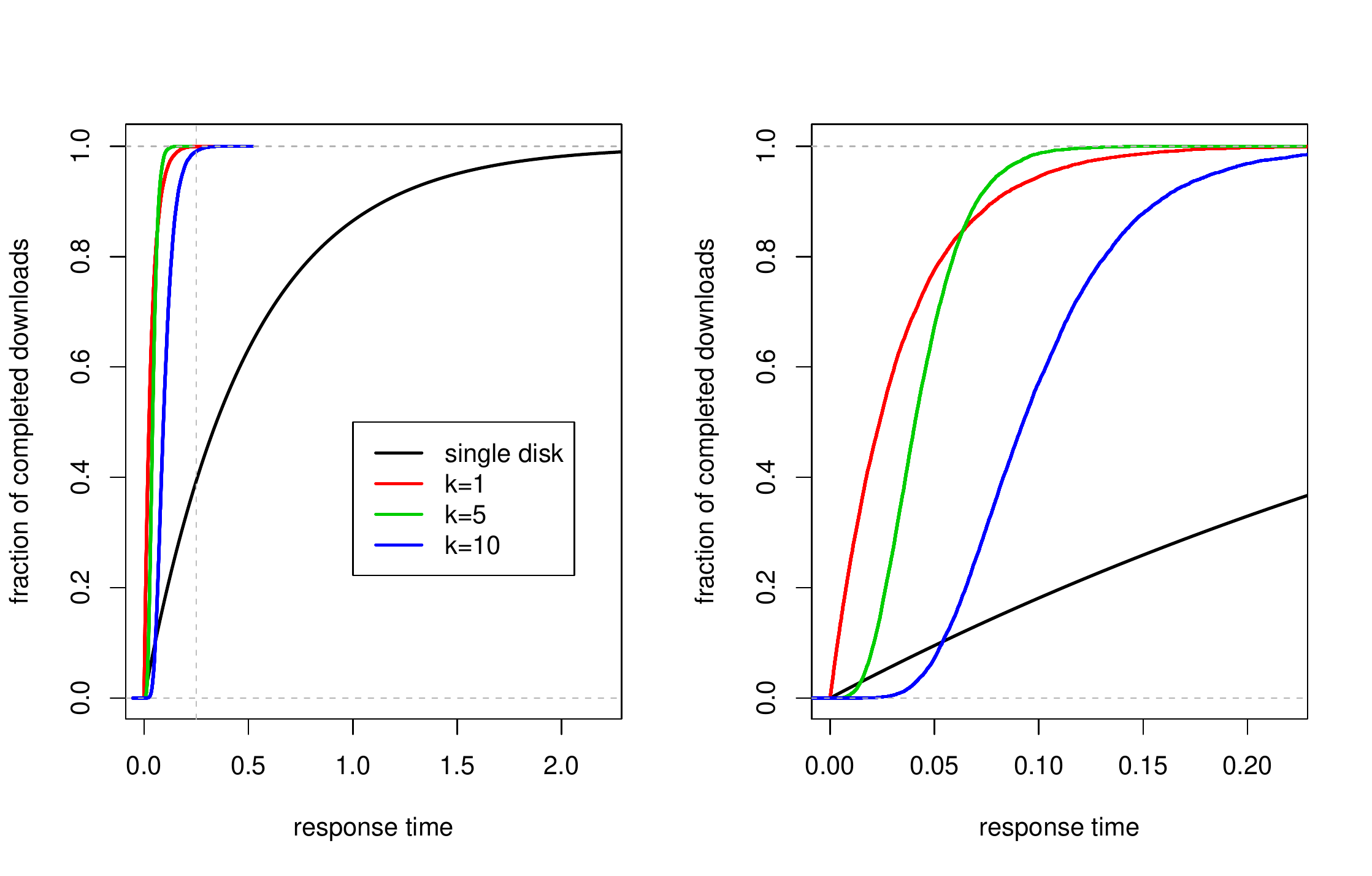}};
\end{tikzpicture}
\end{center}
\vspace{-0.5cm}

\begin{tikzpicture} [scale=0.35]
\def\myangle{360};
\def\dr{25pt}
\def\hr{3pt}
\def\tr{5pt}
\foreach \n in {0} {
\coordinate (A) at (\n,0);
\draw [black!30, - , fill] (A) -- ($(A)+(\dr,0)$)  arc (0:360:\dr) -- cycle;
\draw [black, - , fill] (A) -- ($(A)+(\dr,0)$)  arc (0:\myangle:\dr) -- cycle;
\draw [white, fill] (A) circle (\tr);
\draw [black!50] (A) circle (\hr);
}
\draw (1.5*\dr,0) node [right] {$\leftarrow$ \small{single disk baseline -- unit storage}};
\end{tikzpicture}

\begin{tikzpicture} [scale=0.35]
\def\myangle{36};
\def\dr{25pt}
\def\hr{3pt}
\def\tr{5pt}
\foreach \n in {0,2,4,6,8,10,12,14,16,18} {
\coordinate (A) at (\n,0);
\draw [black!30, - , fill] (A) -- ($(A)+(\dr,0)$)  arc (0:360:\dr) -- cycle;
\draw [blue, - , fill] (A) -- ($(A)+(\dr,0)$)  arc (0:\myangle:\dr) -- cycle;
\draw [white, fill] (A) circle (\tr);
\draw [black!50] (A) circle (\hr);
}
\draw ($(A)+(1.5*\dr,0)$) node [right] {$\leftarrow$ \small{the same total storage}};
\end{tikzpicture}

\begin{tikzpicture} [scale=0.35]
\def\myangle{72};
\def\dr{25pt}
\def\hr{3pt}
\def\tr{5pt}
\foreach \n in {0,2,4,6,8,10,12,14,16,18} {
\coordinate (A) at (\n,0);
\draw [black!30, - , fill] (A) -- ($(A)+(\dr,0)$)  arc (0:360:\dr) -- cycle;
\draw [green, - , fill] (A) -- ($(A)+(\dr,0)$)  arc (0:\myangle:\dr) -- cycle;
\draw [white, fill] (A) circle (\tr);
\draw [black!50] (A) circle (\hr);
}
\draw ($(A)+(1.5*\dr,0)$) node [right] {$\leftarrow$ \small{double total storage}};
\end{tikzpicture}

\begin{tikzpicture} [scale=0.35]
\def\myangle{360};
\def\dr{25pt}
\def\hr{3pt}
\def\tr{5pt}
\foreach \n in {0,2,4,6,8,10,12,14,16,18} {
\coordinate (A) at (\n,0);
\draw [black!30, - , fill] (A) -- ($(A)+(\dr,0)$)  arc (0:360:\dr) -- cycle;
\draw [red, - , fill] (A) -- ($(A)+(\dr,0)$)  arc (0:\myangle:\dr) -- cycle;
\draw [white, fill] (A) circle (\tr);
\draw [black!50] (A) circle (\hr);
}
\draw ($(A)+(1.5*\dr,0)$)  node [right] {$\leftarrow$ \small{$10\times$ increase in storage}};
\end{tikzpicture}
\caption{CDFs of the response time of $(10, k)$ fork-join systems, and the required storage\label{fig:flexible_storage_cdf}  }
\end{figure*}
%---------------- Yanpei Begins Apr. 15 ----------------------

In this section we present numerical results demonstrating the fundamental trade-off between storage and response time of the $(n,k)$ fork-join system. We also compare the response time of the $(n,k)$ fork-join system to the power-of-$d$ and LWL assignment policies introduced in Section~\ref{subsec:power_of_d_intro}. % job assignment policies. %{\bf [Suggest change to ``we also compare the response time of ... to the power-of-d request assignment first proposed in ..." -- Yanpei]}

The expected download time of the file can be reduced in two ways 1) by increasing the total storage, or the storage expansion $n/k$ per file, and 2) by increasing the number $n$ of disks used for file storage. Both the total storage and the number of disks could be a limiting factor in practice. We first address the scenario where the number of disks $n$ is kept constant, but the storage expansion changes from $1$ to $n$ as we choose $k$ from $1$ to $n$. We then study the scenario where the storage expansion factor $n/k$ is kept constant, but the number of disks varies.

%For example, the $(4,2)$ and $(10,5)$ fork-join sysems both provide a storage expansion of $2$, but the former uses $4$ and the latter $10$ disks, and thus their download times behave differently. 

\subsubsection{Flexible Storage Expansion \& Fixed Number of Disks}

Fig.~\ref{fig.Tnk_sim} is a plot of mean response versus $k$ for a fixed number of disks $n$. Note that as we increase $k$, the total storage $n/k$ used decreases as shown in Fig.~\ref{fig:not_tight_bounds}. When we increase $k$, two factors affect the mean response time $T_{(n,k)}$ in opposite ways: 1) As $k$ increases the storage per disk reduces which reduces mean response time. 2) With higher $k$ we have to wait for more nodes to finish service for the job to exit the system. Hence we lose the diversity benefit of coding, which results in an increase in the mean response time.

In Fig.~\ref{fig.Tnk_sim} we observe that when $\mu = 1$ or $3$, the second factor dominates causing the mean response time $T_{(n,k)}$ to strictly increase with $k$. At lower service rate $\mu = \frac{1}{8}$ shown in Fig.~\ref{fig:verybad_bounds}, the mean response time first decreases, and then increases with $k$. At small $k$ (e.g.~$k = 1$), the per-node service time $1/k\mu$ becomes large, outweighing the benefit of waiting for just $k$ nodes to finish service. At large $k$, waiting for many nodes to response outweighs the fast $1/k\mu$ service time. Due to this phenomenon, there is an optimal $k$ that minimizes the mean response time.

%Also note that in Fig.~\ref{fig:not_tight_bounds} both upper and lower bounds become not as tight as the ones in Fig.~\ref{fig:nksim}. {\bf[Mention stability condition. For deterministic service time $k=1$ with $\mu = \nicefrac{1}{8}$ is not stable, but it is with exponential service time]}

In addition to small mean response time, ensuring quality-of-service to the user may also require that the probability of exceeding some maximum tolerable response time to be small. Thus, we study the cumulative distribution function (CDF) of the response time for different values of $k$ for a fixed $n$. In Fig.~\ref{fig:flexible_storage_cdf} we plot the CDF of the response time with $k= 1, 2 , 5, 10$ for fixed $n=10$. The arrival rate and service rate are $\lambda = 1$ and $\mu = 3$ as defined earlier. For $k=1$, the CDF is represents the minimum of $n$ exponential random variables, which is also exponentially distributed.
%The peak of the PDF shifts to the right as $k$ increases.
%\vspace{-1cm}

%{\bf [I just realized that there are two things regarding this figure 1) the dashed curves in the legend does not match the curve style in the graph 2) the legend does not match the disks diagram right below -- Yanpei]}
%\end{example}

The CDF plot can be used to design a storage system that gives probabilistic bounds on the response time. For example, if we wish to keep the response time below $0.1$ seconds with probability at least $0.75$, then the CDF plot shows that $k=5,10$ satisfy this requirement but $k=1$ does not. The plot also shows that at $0.4$ seconds, $100\%$ of requests are complete in all fork-join systems, but only $50\%$ are complete in the single-disk case.

\subsubsection{Flexible Number of Disks \& Fixed Storage Expansion }
Next, we take a different viewpoint and analyze the benefit of spreading the content across more disks while using the same total storage space. Fig.~\ref{fig.paretosim_fixratio} plots the bounds (\ref{eqn:upper_bnd}) and (\ref{eqn:lower_bnd}) on the mean response time $T_{(n,k)}$ versus $k$ while keeping constant code rate $k/n= 1/2$, for the $(n,k)$ fork-join system with $\lambda = 1$ and three different values of $\mu$. %{\bf [The plot by Emina is not correct (perhaps mislabeled -- Yanpei]}
\begin{figure}[ht]
\centering
\includegraphics[scale=0.45]{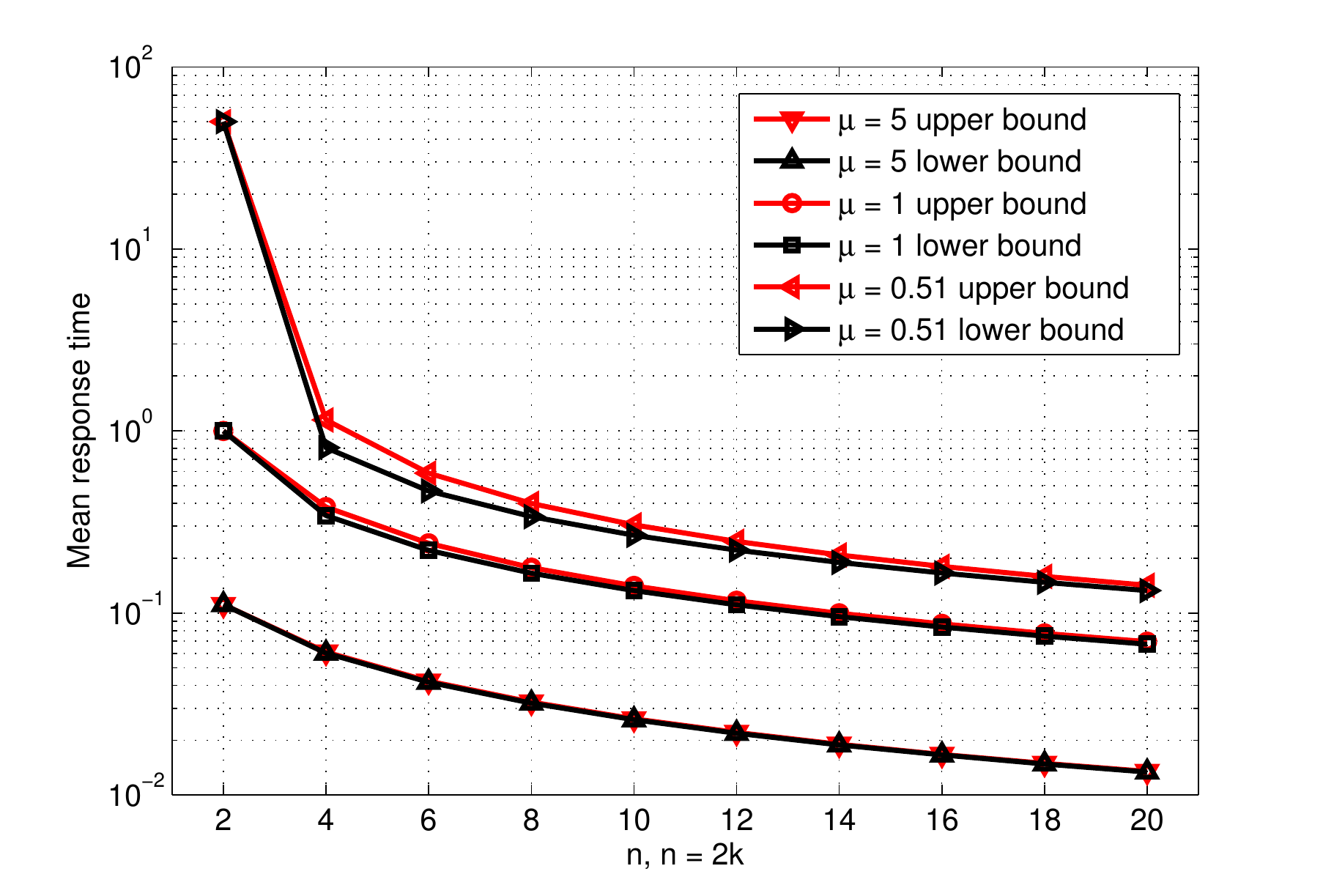}
\caption{Mean response time upper and lower bounds on the mean response time $T_{(n, n/2)}$ for $\lambda = 1$ and three different service rates $\mu$. Due to the diversity advantage of more disks, $T_{(n, n/2)}$ reduces with $n$.}
\label{fig.paretosim_fixratio}
\end{figure}
For these parameter values the bounds are tight, and can be used for analysis in place of simulations.

We observe that the mean response time $T_{(n,k)}$ reduces as $k$ increases because we get the diversity advantage of having more disks.
The reduction in $T_{(n,k)}$ happens at the higher rate for small values of $k$ and $\mu$.
For heavy-tailed distributions (e.g.\ Pareto, cf.\ Sec.~\ref{sec:extensions}), the benefit that comes from diversity is even larger.

$T_{(n,k)}$ approaches zero as $n \rightarrow \infty$ for a fixed storage expansion $n/k$.
This is because we assumed that service rate of a single disk is $k \mu$ since the $1/k$ units of the content $F$ is stored on one disk. However, in practice the mean service time $1/k\mu$ will not go zero as reading each disk will need some non-zero setup time in completing each task, irrespective of the amount of data read from the disk. In Section~\ref{sec:extensions} we will see how this setup time affects the delay-storage trade-off.

In order to understand the response time better, we plot in Fig.~\ref{fig:fixed_storage_cdf} the CDF for different values of $k$ for a fixed ratio $k/n = 1/2$. Again we observe that the diversity of increasing number of disks $n$ helps to reduce the response time.

\begin{figure}[ht]
\begin{center}
%\begin{minipage}[c]{0.6\linewidth}
%\begin{center}
\includegraphics[scale=0.4]{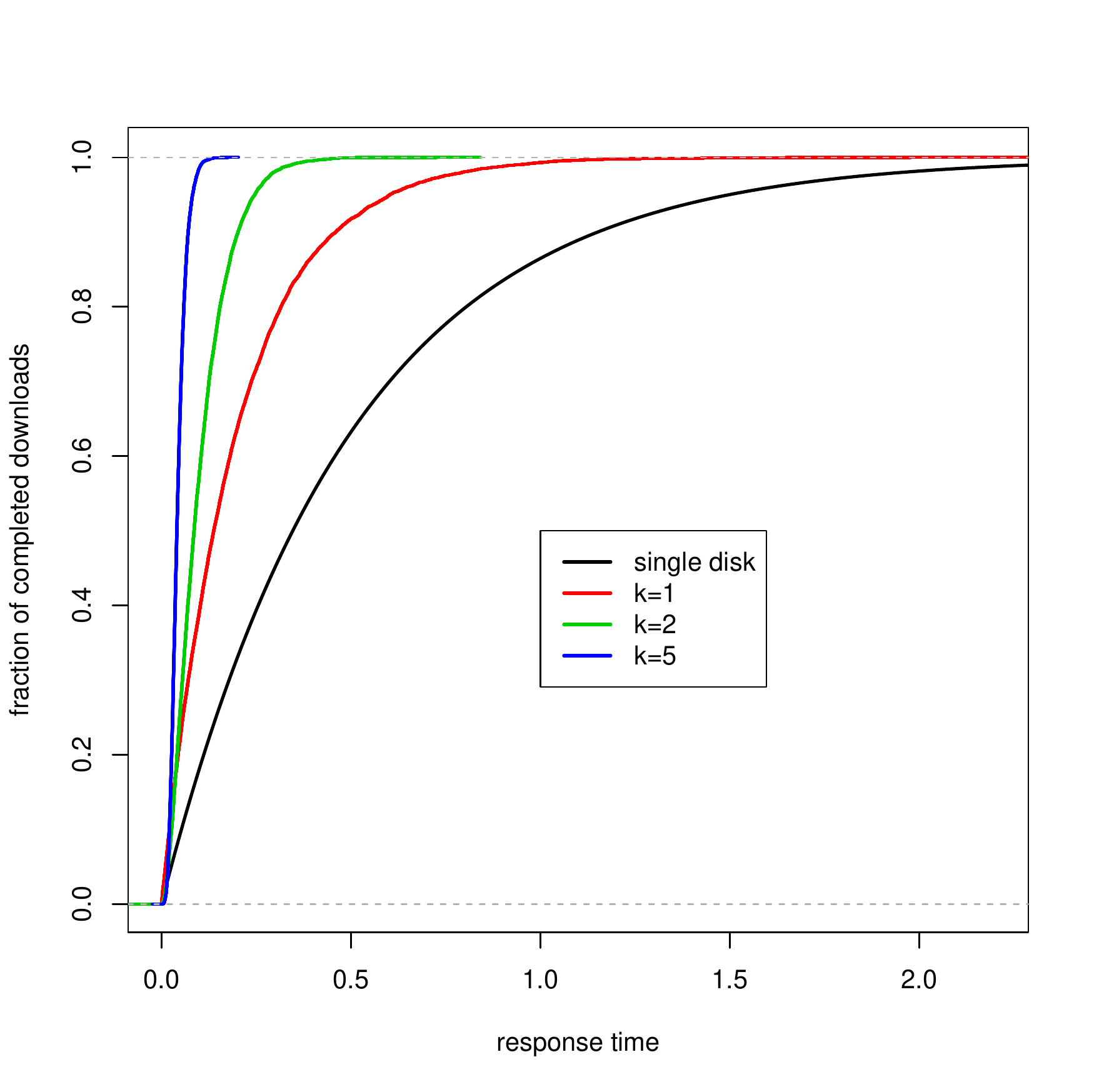}
%\end{center}
\end{center}
\caption{CDFs of the response time of $(n, k = n/2)$ fork-join systems, and the required storage\label{fig:fixed_storage_cdf}}
\end{figure}

\subsubsection{Comparison with Power-of-$d$ Assignment}
%\textcolor{red}{How about moving the general discussion about LWL and power of 2 choices to Sec.~\ref{sec:main_idea}? There we first introduce our
%idea and then say that in multi-queue systems, many routing policies are possible, all premised on entire content availability at any server. This assumption is natural and comes at no cost when e.g., dealing with computing jobs, but requires storing the entire content at each server.}

\begin{figure}[ht]
\centering
\includegraphics[scale=0.45]{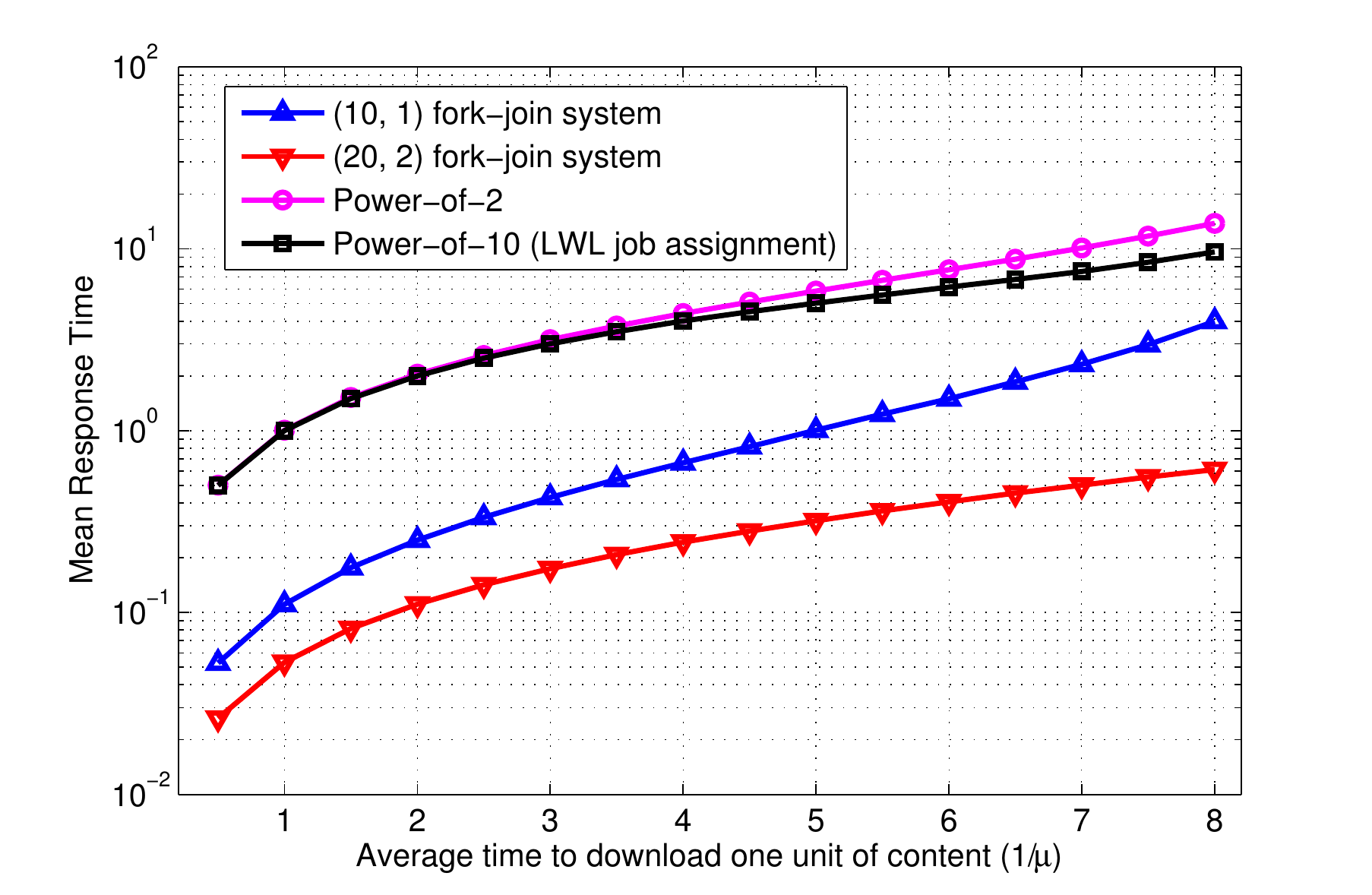}
\caption{For $\lambda = 1$ and the same amount of total storage used  ($10$ units), the fork-join system has lower mean response time than the corresponding power-of-$d$ and LWL assignment policies.}
\label{fig.powerofd_nksim}
\end{figure}

%\textcolor{red}{Then here, we could just make a case that even if storage space is not an issue, fork join results in a shorter download time.
%It would be good to separate the time savings due to the decreased $\mu$ and due to diversity.}

We now compare the mean response time of the $(n,k)$ fork-join system with power-of-$d$ and least-work-left (LWL) job assignment introduced in Section~\ref{sec:main_idea}. Recall that for each incoming request, the power-of-$d$ policy assigns a request to the node with the least-work-left from among $d$ uniformly selected nodes. Fig.~\ref{fig.powerofd_nksim} is a plot of the mean response time versus $1/\mu$, the average time taken to download one unit of content. It compares the $(n,k)$ fork-join system which uses $n/k$ units of total storage with the power-of-$d$ and LWL assignment policies with the entire content (one unit) replicated on the $n/k$ disks. Thus, all the systems shown in Fig.~\ref{fig.powerofd_nksim} use the same total storage space $n/k = 10$ units. 

We observe in Fig.~\ref{fig.powerofd_nksim} that the fork-join system outperforms the power-of-$d$ and LWL assignment policies. This is because as we saw in Fig.~\ref{fig.paretosim_fixratio}, when we increase $n$ and $k$ while keeping the ratio $n/k$ (the total storage) fixed, the mean response time of the $(n,k)$ fork-join system decreases. That is, the diversity advantage dominates over the slowdown due to waiting for more nodes to finish service. Thus, for large enough $n$, the $(n,k)$ fork-join system outperforms the corresponding power-of-d scheme that uses the same storage space $n/k$ units. 

%We have edited it the explanation of Fig. 7 and hope that the current version is clear. due to two reasons: 1) since only $1/k$ units of content are stored and read in parallel from the disks, the service rate to a head-of-the-queue request is $k \mu$ for the fork-join system as compared to $\mu$ for the power-of-$d$ and LWL policies. 2) the diversity advantage of sending redundant requests and waiting for only $k$ out of $n$ content blocks to be read.

There are other practical issues that are not considered in Fig.~\ref{fig.powerofd_nksim}. For instance, in the $(n, k)$ fork-join system there are communication costs associated with forking jobs to $n$ nodes and costs of decoding the MDS coded blocks. On the other hand, the power-of-$d$ assignment system requires constant feedback from the nodes to determine the work left at each node.

%Fig.~\ref{fig.powerofd_nksim} shows how the fork-join system outperforms power-of-$d$ job assignment in terms of mean response time. Note that in Fig.~\ref{fig.powerofd_nksim} all simulated systems have the same storage size, i.e., $10$ units of data. The system parameters are job arrival rate $\lambda = 1$. The benefits of fork-join system comes from the diversity it gains by sending requests to multiple nodes. Also, we note that the power-of-$2$ policy is only slightly worse than the LWL policy.

\section{Generalizing the Service Distribution}
\label{sec:extensions}

The theoretical analysis and numerical results so far assumed a specific service time distribution at each node -- we considered the exponential distributions. In this section we present some results by generalizing the service time distribution. In Section~\ref{subsec:fj_general_service_time} we extend the upper bound to general service time distributions. We present numerical results for heavy-tailed and correlated service times in Section~\ref{subsec:heavy_tailed_service} and Section~\ref{subsec:correlated_service} respectively. 

% Extensions

%--------- Yanpei Added ----------------------------------------------
\subsection{General Service Time Distribution}
\label{subsec:fj_general_service_time}
%
%The theoretical analysis and numerical results so far assumed a specific service time distribution at each node -- we considered the exponential distributions. However, 

In several practical scenarios the service distribution is unknown. We present an upper bound on the mean response time for such cases, only using the mean and the variance of the service distribution. Let $X_1, X_2, \ldots, X_n$ be the i.i.d random variables representing the service times of the $n$ nodes, with expectation $\expec[X_i] = \frac{1}{\mu'}$ and variance $\var[X_i] = \sigma^2$ for all $i$. 

\begin{thm}[Upper Bound with General Service Time]
The mean response time $T_{(n,k)}$ of an $(n, k)$ fork-join system with general service time $X$ such that $\expec[X] = \frac{1}{\mu'}$ and $\var[X] = \sigma^2$ satisfies
\begin{align}
T_{(n,k)} \leq \frac{1}{\mu'} &+ \sigma \sqrt{\frac{k-1}{n-k+1}} \nonumber \\
&+ \frac{\lambda \left[ \left ( \frac{1}{\mu'} + \sigma \sqrt{\frac{k-1}{n-k+1}} \right)^2 + \sigma^2 C(n,k) \right]}{2 \left[ 1- \lambda \left( \frac{1}{\mu'} + \sigma \sqrt{\frac{k-1}{n-k+1}} \right ) \right ]} \label{eqn:upper_bnd_gen}.
\end{align}
\end{thm}
\begin{proof}
The proof follows from Theorem~\ref{thm:upper_bnd} where the upper bound can be calculated using $(n, k)$ split-merge system and Pollaczek-Khinchin formula (\ref{eqn:pollac_khin}). Unlike the exponential distribution, we do not have an exact expression for $S$, i.e., the $k^{th}$ order statistic of the service times $X_1, \,\, X_2, \,\, \cdots X_n$. Instead, we use the following upper bounds on the expectation and variance of $S$ derived in \cite{ArnoldGroeneveld} and\cite{Papadatos}.
\begin{align}
\expec[S] &\leq \frac{1}{\mu'} + \sigma \sqrt{\frac{k-1}{n-k+1}}, \label{eqn:upperbound_exp} \\
\var[S] &\leq C(n,k) \sigma^2. \label{eqn:upperbound_var}
\end{align}

The proof of (\ref{eqn:upperbound_exp}) involves Jensen's inequality and Cauchy-Schwarz inequality. For details please refer to \cite{ArnoldGroeneveld}. The constant $C(n,k)$ depends only on $n$ and $k$, and can be found in the table in \cite{Papadatos}. Holding $n$ constant, $C(n,k)$ decreases as $k$ increases.  The proof of (\ref{eqn:upperbound_var}) can be found in \cite{Papadatos}.

Note that (\ref{eqn:pollac_khin}) strictly increases as either $\expec[S]$ or $\var[S]$ increases. Thus, we can substitute the upper bounds in it to obtain the upper bound on mean response time (\ref{eqn:upper_bnd_gen}).
\end{proof}

Regarding the lower bound, we note that our proof in Theorem~\ref{thm:lower_bnd} cannot be extended to this general service time setting. The proof requires memoryless property of the service time, which does not necessary hold in the general service time case.

%We consider the Pareto distribution for service time in the simulation results presented in Sec.~\ref{subsec:num_results_n_k} . {\bf [added -- Yanpei]- edited - Gauri}

\subsection{Heavy-tailed Service Time}
\label{subsec:heavy_tailed_service}

In many practical systems the service time has a heavy-tail distribution, which means that there is a larger probability of getting very large values. More formally, a random variable $X$ is said to be heavy-tail distribution if its tail probability is not exponentially bounded and $\lim_{x \rightarrow \infty} e^{\beta  x} \Pr(X>x) = \infty$ for all $\beta > 0$. 
We consider the Pareto distribution which has been widely used to model heavy-tailed jobs in existing literature (see for example \cite{CrovellaBestavros, FaloutsosFaloutsos}). The Pareto distribution is parametrized is parametrized by scale parameter $x_m$ and shape parameter $\alpha$ and its cumulative distribution function is given by,
\begin{equation}
F_X(x) = 
\begin{cases}
1 - \left( \frac{x_m}{x} \right)^{\alpha} & \text{for } x \geq x_m \\
0 & \text{for } x < x_m 
\end{cases}
\end{equation}
A smaller value of $\alpha$ implies a heavier tail. In particular, when $\alpha = \infty$ the service time becomes deterministic and when $\alpha \leq 1$ the service time becomes infinite. In \cite{CrovellaBestavros} Pareto distribution with $\alpha = 1.1$ was reported for the sizes of files requested from websites.

% {\bf [This is the first time the $X$ notation appears in the paper. Suggest we introduce early in Section III and remind the readers again here -- Yanpei]-fixed -Gauri}

\begin{figure}[t]
\centering
\includegraphics[scale=0.45]{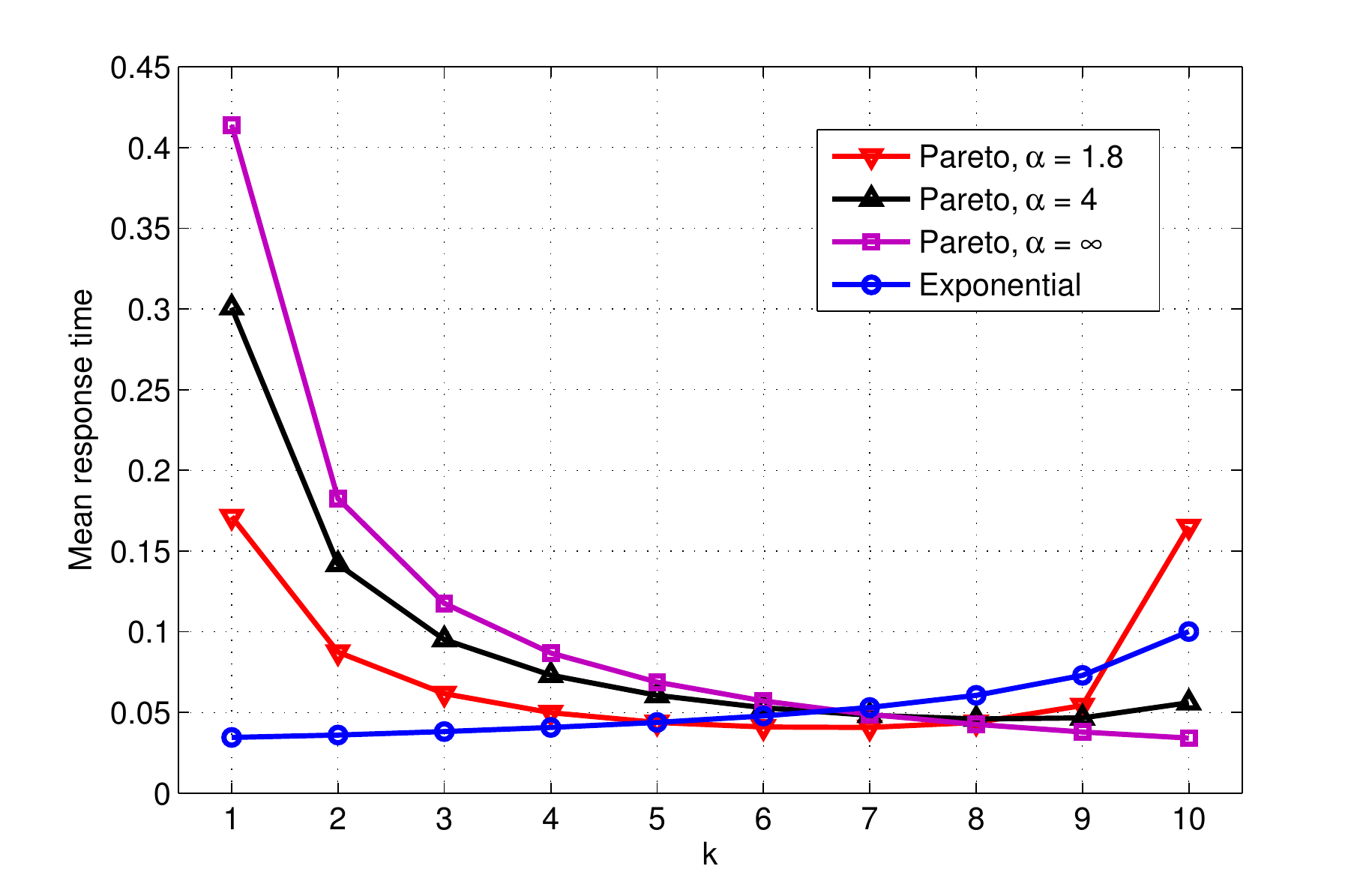}
\caption{Mean response time $T_{(10, k)}$ of different service time distributions. $\lambda = 1$ and $\mu = 3$. For more heavy-tailed (smaller $\alpha$) distributions, the increase in mean response time with $k$ becomes dominant since we have to wait for more nodes to finish service.}
\label{fig.paretonksim}
\end{figure}
In Fig.~\ref{fig.paretonksim} we plot the mean response time $T_{(n,k)}$ versus $k$ for $n=10$ disks, for arrival rate $\lambda = 1$ and service rate $\mu = 3$ for the exponential and Pareto service distributions. Each disk stores $1/k$ units of data and thus the service rate of each individual queue is $\mu' = k \mu$. For a given $k$, all distributions have the same mean service time $1/k\mu$. We observe that as the distribution becomes more heavy-tailed (smaller $\alpha$), waiting for more nodes (larger $k$) to finish results in an increase in mean response time which outweighs the decrease caused by smaller service time $1/k\mu$. For smaller $\alpha$, the optimal $k$ decreases because the increase in mean response time for larger $k$ is more dominant. %{\bf [Not sure how this follows from the previous sentence -- Yanpei] -- fixed- Gauri}

\subsection{Correlated Service Times}
\label{subsec:correlated_service}
%{\color{red} [The $\alpha$ used here is already used for Pareto distribution. Considering changing to $\gamma$ ($\beta$ is also used) -- Yanpei]}
Thus far we have considered that the $n$ tasks of a job have independent service times. We now analyze how the correlation between service times affects the mean response time of the fork-join system. In practice the correlation between service times could be because the service time is proportional to the size of the file being downloaded. We model the correlation by considering that the service time of each task is $\delta X_d + (1- \delta) X_{r,i}$, a weighted sum of two independent exponential random variables $X_d$ and $X_{r,i}$ both with mean $1/k\mu$. The variable $X_d$ is fixed across the $n$ queues, and $X_{r,i}$ is the independent for the queues $1 \leq i \leq n$. The weight $\delta$ represents the degree of correlation between the service times of the $n$ queues. When $\delta = 0$, the system is identical to the original $(n,k)$ fork-join system analyzed in Section~\ref{sec:n_k_fork_join}.
The mean response time $T'_{(n,k)}$ of the $(n,k)$ fork-join system with service time distribution as described above is,
\begin{align}
T'_{(n,k)} &= \delta \expec[X_d] + (1-\delta) T_{(n,k)}, \label{eqn:job_dep_serv_time} \\
&= \frac{\delta}{k \mu} + (1-\delta) T_{(n,k)},\nonumber
\end{align}
where in $T_{(n,k)}$ is the response time with independent exponential service times analyzed in Section~\ref{sec:n_k_fork_join}. 
%In \eqref{eqn:job_dep_serv_time}, the job-dependent part of the response time $X_d$ is fixed for all queues and hence contributes to an additive term. Since the second part $X_{r,i}$ is independent across queues, its analysis is identical to that of $T_{(n,k)}$ presented in Section~\ref{subsec:bounds}.
%
Fig.~\ref{fig:job_dep_sim} shows the trade-off between mean response time and $k$ for weight $\delta = 0, 0.5,$ and $1$. When $\delta$ is $0$, coding provides diversity in this regime and gives faster response time for smaller $k$ as we already observed in Fig.~\ref{fig.Tnk_sim}. As the correlation between service times increases we lose the diversity advantage provided by coding and do not get fast response for small $k$. Note that for $\delta = 1$, there is no diversity advantage and the decrease in response time with $k$ is only because of the fact that each disk stores $1/k$ units of data.
%
%When $\delta$ is close to $0$, the service time of the $n$ queues are almost independent. Coding provides diversity in this regime and gives faster response time for smaller $k$ as we already observed in Fig.~\ref{fig.Tnk_sim}. On the other hand when $\alpha$ is close to $1$, the service times of the queues are highly correlated and proportional to the job size. Consequently, we lose the diversity advantage provided by coding, and the response time goes down with $k$ only because of the fact that each disk stores $1/k$ units of data.
%
\begin{figure}[hbt]
\centering
\includegraphics[scale=0.45]{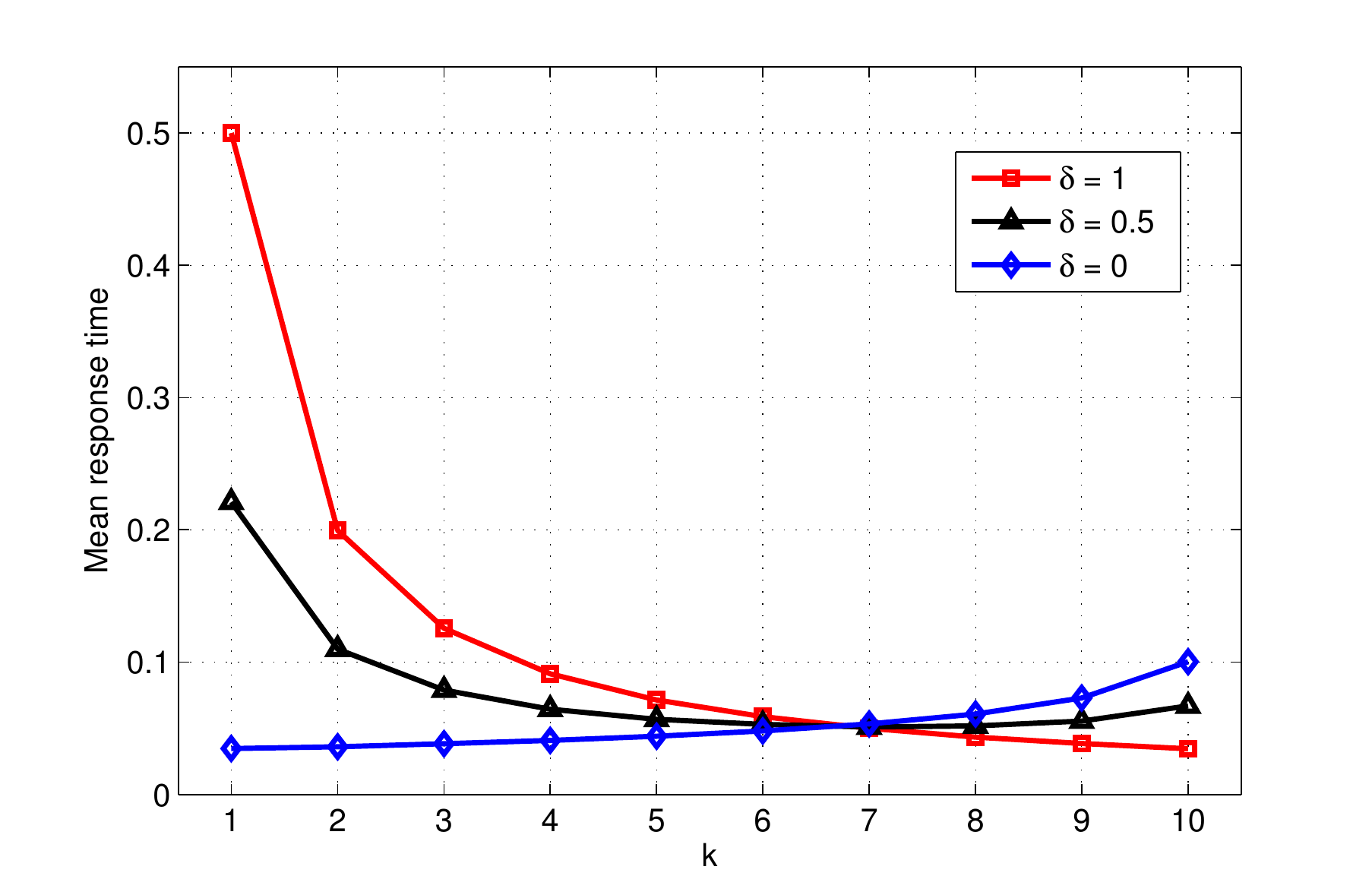}
\caption{ \label{fig:job_dep_sim} Mean response time of $T'_{(n,k)}$ of $(10, k)$ fork-join systems with job dependent service time distribution for $\lambda = 1$, $\mu = 3$, and $\delta = 0, 0.5 ,1$. As $\delta$ increases, the service times are more correlated and we lose the diversity advantage of coding.}
\end{figure}

\section{The $(m,n,k)$ Fork-join System}
\label{sec:m_n_k_fork_join}

In a distributed storage with a large number of disks $m$, having an $(m,k)$ fork-join system would involve large signaling overhead of forking the request to all the $m$ disks, and high decoding complexity. The decoding complexity is high even with small $k$ because it depends on the field size, which is a function of $m$ in standard codes such as Reed-Solomon codes. Hence, we propose a system where we divide the $m$ disks into $g = m/n$ groups of $n$ disks each, which act as independent $(n,k)$ fork-join systems. In Section~\ref{subsec:sys_model_m_n_k} we give the system model and analyze the mean response time of the $(m,n,k)$ fork-join system. In Section~\ref{subsec:num_results_m_n_k} we present numerical results comparing the mean response time with different policies of assigning an incoming request to one of the groups.

%In a distributed storage with a large number of disks $m$, having an $(m,k)$ fork-join system and forking requests to all disks would involve significant MDS code decoding complexity and signaling overhead. We propose a system where we divide the disks into groups which act as independent $(n,k)$ fork-join systems. In Section~\ref{subsec:sys_model_m_n_k} we give the system model and analyze the mean response time of the $(m,n,k)$ fork-join system. In Section~\ref{subsec:num_results_m_n_k} we present numerical results comparing the mean response time with different policies of assigning an incoming request to one of the groups.

\subsection{Analysis of Response Time}
%{\bf [Suggest removing the subtitle ``system model", reason as before -- Yanpei]}
Consider a distributed storage system with $m$ disks. We divide then into $ g = m/n$ groups of $n$ disks each as shown in Fig.~\ref{fig:groups}. We refer to this system as the $(m,n,k)$ fork-join system, formally defined as follows.

\begin{defn}[The $(m,n,k)$ fork-join system]
\label{defn:m_n_k_fork_join}
An $(m,n,k)$ fork-join system consists of $m \ge n$ disks partitioned into $g = m/n$ groups with $n$ disks each. An incoming download request is assigned to one of the $g$ groups according to some policy (e.g., uniformly at random). Each group behaves as an independent $(n,k)$ fork-join system described in Definition~\ref{defn:n_k_fork_join}.
\end{defn}
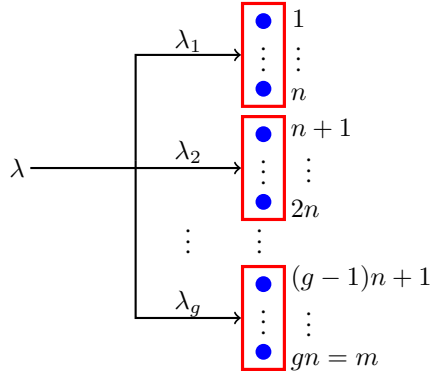
\begin{figure}[ht]
\centering
\begin{tikzpicture}[scale=0.70]
\begin{small}
\matrix [draw=red,very thick,column sep=1cmx, matrix anchor=west] (group1) at (0,0)
{
\node {}; \fill[blue] (0,0) circle (3pt);\\[-1mm]
\node {\vdots};\\
\node {}; \fill[blue] (0,0) circle (3pt);\\
};
[column/.style={anchor=base west}]
\matrix [column sep=1cmx, matrix anchor=west] at (0.6,-0.0)
{
\node {$1$};\\[-1.7mm]
\node {$\vdots$};\\
\node {{$n$}};\\ %\textcolor{red}
};
\matrix [draw=red, very thick, column sep=1cmx, matrix anchor=west] (group2) at (0,-2.15)
{
\node {}; \fill[blue] (0,0) circle (3pt);\\[-1mm]
\node {\vdots};\\
\node {}; \fill[blue] (0,0) circle (3pt);\\
};
\matrix [column sep=1cmx, column 1/.style={anchor=base west}, matrix anchor=west] at (0.6,-2.15)
{
\node {$n+1$};\\[-1.7mm]
\node {$\;\;\vdots$};\\
\node {$2n$};\\
};
\node at (0.35,-3.4) {$\vdots$};\node at (-1,-3.4) {$\vdots$};
\matrix [draw=red,very thick,column sep=1cmx, matrix anchor=west] (groupg) at (0,-5)
{
\node {}; \fill[blue] (0,0) circle (3pt);\\[-1mm]
\node {$\vdots$};\\
\node {}; \fill[blue] (0,0) circle (3pt);\\
};
\matrix [column sep=1cmx, column 1/.style={anchor=base west}, matrix anchor=west] at (0.6,-5)
{
\node {$(g-1)n+1$};\\[-1.7mm]
\node {$\;\;\vdots$};\\
\node {{$gn=m$}};\\ %\textcolor{red}
};
\draw [thick,->] (-2,-2) |- (group1.west);
\draw [thick,->] (-4,-2.15) -- (group2.west);
\draw [thick,->] (-2,-2) |- (groupg.west);
\node at (-4.25,-2.15){$\lambda$};
\node at (-1,0.27) {$\lambda_1$};
\node at (-1,-1.85) {$\lambda_2$};
\node at (-1,-4.73) {$\lambda_g$};
\end{small}
\end{tikzpicture}
\caption{ The $(m,n,k)$ fork-join system with incoming service requests split among $g=m/n$ fork-join systems.\label{fig:groups} }
\end{figure}

We can extend Lemma~\ref{defn:n_k_fork_join} to find a necessary condition for the stability of the $(m,n,k)$ fork-join system, in terms of the arrival rate $\lambda_i$ to each group $i$, for $1 \leq i \leq g$.

\begin{lem}[Stability of $(m,n,k)$ fork-join system]
\label{lem:stability_m_n_k}
For the $(m,n,k)$ fork-join system to be stable, the rate of arrival of requests $\lambda_i$ to group $i$ and the service rate $\mu' = k \mu $ per node must satisfy
\begin{align}
\lambda_i &< n\mu, \quad \forall~1 \leq i \leq g.\label{eqn:stability_m_n_k}
\end{align}
\end{lem}
\begin{proof}
Since each group behaves as an independent $(n,k)$ fork-join system we can apply the condition of stability in Lemma~\ref{lem:stability_n_k} with $\lambda$ replaced by $\lambda_i$. The result follows from this.
\end{proof}
The response time of the $(m,n,k)$ fork-join system depends on the policy used to assign an incoming request to one of the groups. Under the uniform job assignment policy, each incoming request is assigned to a group chosen uniformly at random from the $g$ groups. The Poisson arrival rate to each group is then reduced to $\lambda/g$, and each group is an independent $(n,k)$ fork-join system. Therefore, we can extend the bounds in Theorem~\ref{thm:upper_bnd} and Theorem~\ref{thm:lower_bnd} to the mean response time of the $(m,n,k)$ fork-join system as follows.
% Moreover, with Poisson arrivals and exponential service time, under the uniform assignment policy, each group behaves as an independent FHW system.  -- Removed because it is confusing
\begin{corr}
\label{corr:bounds_m_n_k}
The response time $T_{(m,n,k)}$ of an $(m,n,k)$ fork-join system with uniform group assignment is bounded by \eqref{eqn:upper_bnd} and \eqref{eqn:lower_bnd} with $\lambda$ replaced by $\lambda/g$.
\end{corr}

\label{subsec:sys_model_m_n_k}

\subsection{Numerical Results}
% Results for (m, n, k) system

\begin{figure}[t]
\centering
\includegraphics[scale=0.45]{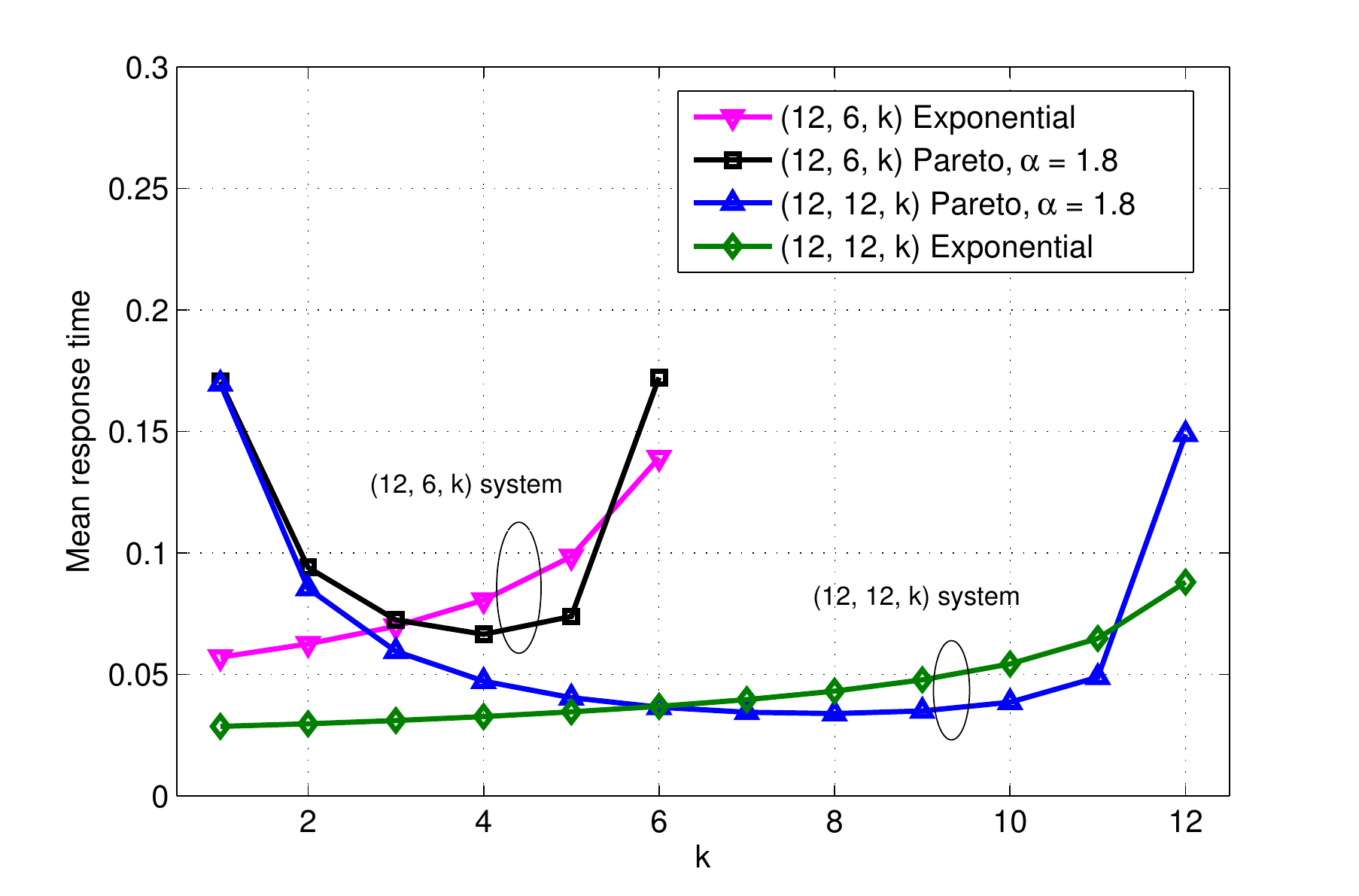}
\caption{Mean response time $T_{(12, n, k)}$ with the exponential and Pareto service distributions, and parameters $\lambda = 1$ and $\mu = 3$. Given $m$ and we would like to find the smallest $n$ and largest $k$ that can achieve a given target response time.}
\label{fig.paretomnksim}
\end{figure}

To reduce the decoding complexity and signaling overhead, an $(m,n,k)$ fork-join system with smaller $n$, and thus more groups $g = m/n$, is preferred. However, reducing $n$ reduces the diversity advantage which could give higher expected download time (cf.\ Fig.~\ref{fig:fixed_storage_cdf}). Thus, there is a delay-complexity trade-off when we vary the number of groups $g$. Moreover, the content has to be replicated at all groups to which its request can be directed. Thus, having a large number of groups, also means increased storage space.

In Fig.~\ref{fig.paretomnksim} we plot the mean response time for $(12, n, k)$ system and uniform group assignment with exponential and Pareto service times. Given the number of disks $m$, we would like to find the smallest $n$, and largest $k$ that can achieve a given target response time. Smaller $n$ means there are less disks per group, and hence less signaling overhead of forking a request to the disks in a group. Larger $k$ is desirable because the total storage space used is $m/k$ units. For exponential service distribution, Fig.~\ref{fig.paretomnksim} shows that diversity of having a large $n$, or smaller $k$ always gives lower response time. But this monotonicity does not hold for the Pareto service time distribution. For example, the $(12,6,3)$ fork-join system with $n=6$ disks per group and $12/3 = 4$ units of storage used, gives lower response time than the $(12,12,2)$ fork-join system with $n=12$ disks per group and total storage $12/2 = 6$ units of storage used. 

We now study the response time of the $(m,n,k)$ fork-join system under three different group assignment policies --  the uniform job assignment policy, where each incoming request is assigned to a group chosen uniformly at random from the $g$ groups, and the power-of-$d$ and least-work-left (LWL) policies introduced in Section~\ref{subsec:power_of_d_intro}. 

 %The power-of-$d$ group assignment policy first chooses $d$ groups out of $g$ groups. Then it assigns the incoming job to the group with the smallest average finishing time across its queues within the group. When $d = g$, the assignment policy reduces to the LWL policy. 
In Fig.~\ref{fig:powerof2} we show a comparison of the response time of the $(20,n,k)$ fork-join system with the uniform, power-of-$d$ and LWL group assignment policies. Request arrival are Poisson with rate $\lambda = 1$ and  service times are exponential with rate $\mu = 1/8$. As expected, the power-of-$d$ assignments give lower response time than the uniform assignment but it is at the cost of receiving feedback about the amount of work left at each node. We again note that power-of-2 policy is only slightly worse than the LWL policy (cf.~Fig.~\ref{fig.powerofd_nksim}). The simulation suggests power-of-$d$ group assignment is a strategy worth considering in actual implementations.

\begin{figure}[t]
\centering
\includegraphics[scale=0.45]{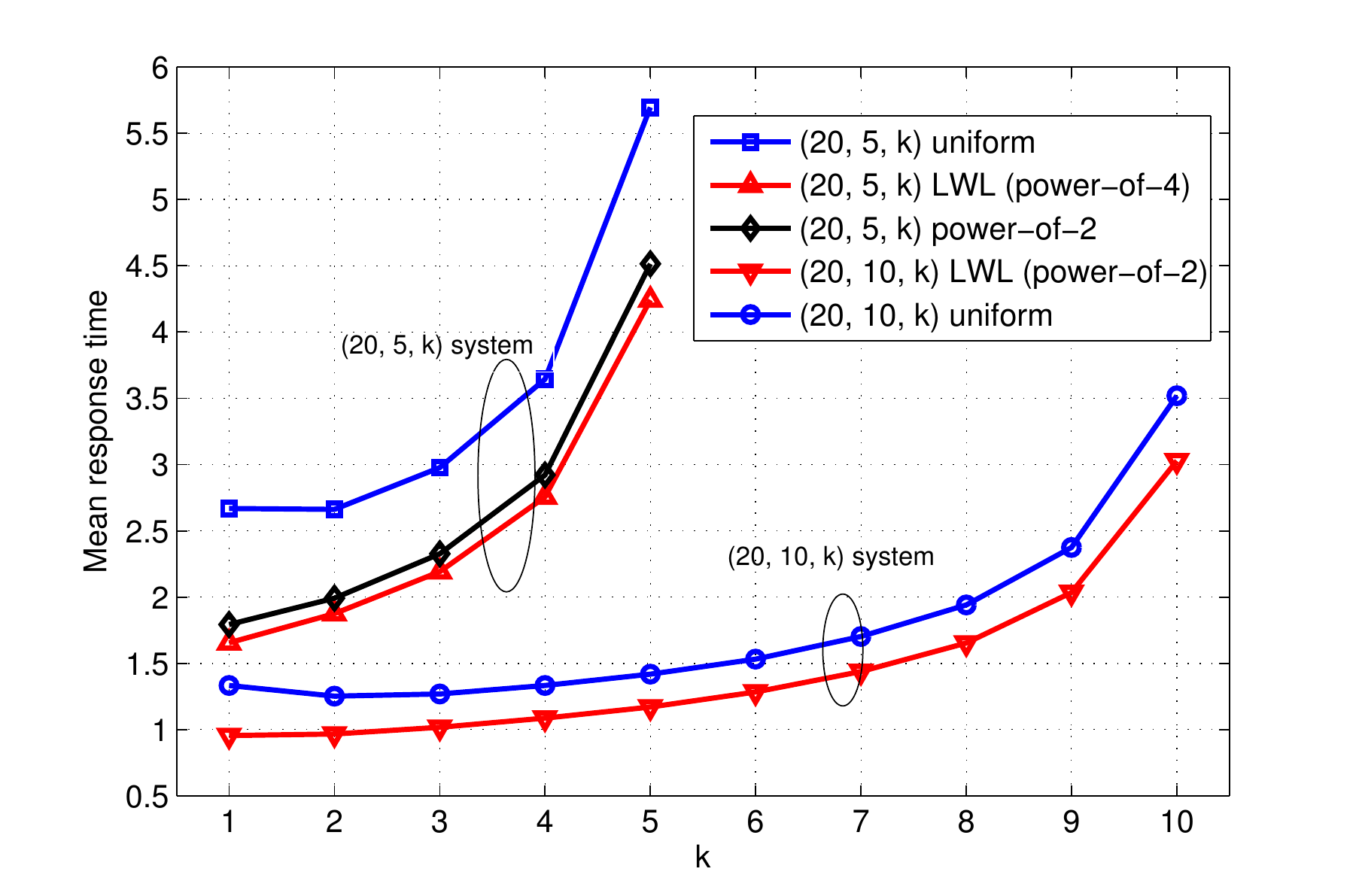}
\caption{ \label{fig:powerof2} Mean response time of $(20, n, k)$ systems, with $\lambda = 1$ and  $\mu = 1/8$ for different group assignment policies. The power-of-$2$ and LWL assignment give faster response time than uniform assignment. }
\end{figure}

\label{subsec:num_results_m_n_k}

\section{Concluding Remarks}
\label{sec:conclu}

\subsection{Major Implications}
%Findings
In this paper we show how coding in distributed storage systems, which has been used to provide reliability against disk failures, also reduces the content download time. We consider that content is divided into $k$ blocks, and stored on $n> k$ disks or nodes in a network. The redundancy is added using an $(n,k)$ maximum distance separable (MDS) code, which allows content reconstruction by reading any $k$ of the $n$ disks. Since the download time from each disk is random, waiting for only $k$ out of $n$ disks reduces overall download time significantly. 

We take a queueing-theoretic approach to model multiple users requesting the content simultaneously. We propose the $(n,k)$ fork-join system model where each request is forked to queues at the $n$ disks. This is a novel generalization of the $(n,n)$ fork-join system studied in queueing theory literature. We analytically derive upper and lower bounds on the expected download time and show that they are fairly tight. To the best of our knowledge, we are the first to propose the $(n,k)$ fork-join system and find bounds on its mean response time. We also extend this analysis to distributed systems with large number of disks, that can be divided into many $(n,k)$ fork-join systems.

%Conclusions
Our results demonstrate the fundamental trade-off between the download time and the amount of storage space. This trade-off can be used for design of the amount of redundancy required to meet the delay constraints of content delivery. We observe that the optimal operating point varies with the service distribution of the time to read each disk. We present theoretical results for the exponential distribution, and simulation results for the heavy-tailed Pareto distribution. %We also show the delay-storage trade-off by considering job-dependent service time, and non-negligible content preparation time. 

\subsection{Future Perspectives}
% Perspectives
Although, we focus on distributed storage here, the results in this paper can be extended to computing systems such as MapReduce \cite{MapReduce} as well as content access networks \cite{VulimiriMichel, Maxemchuk93}. %Our work is part of the big collection of literature already in place that concerns delay diversity trade-off. Some of these works include Multipath TCP \cite{VulimiriMichel}, MapReduce \cite{MapReduce}, and other computer networks \cite{Maxemchuk93}.

There are some practical issues affecting the download time that are not considered in this paper and could be addressed in future work. For instance, the signaling overhead of forking the request to $n$ disks, and the complexity of decoding the content increases with $n$. In practical storage systems, adding redundancy in storage not only requires extra capital investment in storage and networking but also consumes more energy \cite{BostoenMullender}. It would be interesting to study the fundamental trade-off between power consumption and quality-of-service. Finally, note that in this paper we focus on the \emph{read} operation in a storage system. However in practical systems requests entering the system consist of both \emph{read} and \emph{write} operations -- we leave the investigation of the \emph{write} operation for future work.

\bibliographystyle{IEEEtran}
\bibliography{fork_join}

\begin{thebibliography}{10}
\providecommand{\url}[1]{#1}
\csname url@samestyle\endcsname
\providecommand{\newblock}{\relax}
\providecommand{\bibinfo}[2]{#2}
\providecommand{\BIBentrySTDinterwordspacing}{\spaceskip=0pt\relax}
\providecommand{\BIBentryALTinterwordstretchfactor}{4}
\providecommand{\BIBentryALTinterwordspacing}{\spaceskip=\fontdimen2\font plus
\BIBentryALTinterwordstretchfactor\fontdimen3\font minus
  \fontdimen4\font\relax}
\providecommand{\BIBforeignlanguage}[2]{{%
\expandafter\ifx\csname l@#1\endcsname\relax
\typeout{** WARNING: IEEEtran.bst: No hyphenation pattern has been}%
\typeout{** loaded for the language `#1'. Using the pattern for}%
\typeout{** the default language instead.}%
\else
\language=\csname l@#1\endcsname
\fi
#2}}
\providecommand{\BIBdecl}{\relax}
\BIBdecl

\bibitem{EBS}
{Amazon EBS}, \url{http://aws.amazon.com/ebs/}.

\bibitem{GFS}
S.~Ghemawat, H.~Gobioff, and S.-T. Leung, ``The {Google} file system,'' in
  \emph{ACM SIGOPS Op.\ Sys.\ Rev.}, vol.~37, no.~5, 2003, pp. 29--43.

\bibitem{Dropbox}
{Dropbox}, \url{http://www.dropbox.com/}.

\bibitem{Googledoc}
{Google Docs}, \url{http://docs.google.com/‎}.

\bibitem{gauri_yanpei_emina_allerton}
G.~Joshi, Y.~Liu, and E.~Soljanin, ``Coding for fast content download,''
  \emph{Allerton Conf.\ on Commun.\ Control and Comput.}, pp. 326--333, Oct.
  2012.

\bibitem{dimakis07}
A.~G. Dimakis, P.~B. Godfrey, M.~Wainwright, and K.~Ramchandran, ``Network
  coding for distributed storage systems,'' \emph{Proc.\ of IEEE INFOCOM}, pp.
  2000--2008, May 2007.

\bibitem{even_odd}
{M. Blaum, J. Brady, J. Bruck, and J. Menon}, ``{EVENODD: an efficient scheme
  for tolerating double disk failures in RAID architectures},'' \emph{IEEE
  Transactions on Computers}, vol.~44, no.~2, pp. 192--202, 1995.

\bibitem{liskov}
R.~Rodrigues and B.~Liskov, ``High availability in {DHTs}: Erasure coding vs.
  replication,'' \emph{Int.\ Workshop Peer-to-Peer Sys.}, pp. 226--239, Feb.
  2005.

\bibitem{kumar}
K.~V. Rashmi, N.~B. Shah, P.~V. Kumar, and K.~Ramchandran, ``Explicit
  construction of optimal exact regenerating codes for distributed storage,''
  \emph{Allerton Conf.\ on Commun.\ Control and Comput.}, pp. 1243 -- 1249,
  Sep. 2009.

\bibitem{kumar2}
N.~B. Shah, K.~V. Rashmi, P.~V. Kumar, and K.~Ramchandran, ``Interference
  alignment in regenerating codes for distributed storage: necessity and code
  constructions,'' \emph{IEEE Trans.\ Inform.\ Theory}, vol.~58, pp.
  2134--2158, Apr. 2012.

\bibitem{zigzag_codes}
{I. Tamo, Z. Wang and J. Bruck}, ``{Zigzag Codes: MDS Array Codes With Optimal
  Rebuilding},'' \emph{IEEE Transactions on Information Theory}, vol.~59,
  no.~3, pp. 1597--1616, 2013.

\bibitem{ulric_muriel_emina}
{U. Ferner, M. {M\'{e}dard}, and E. Soljanin}, ``Toward sustainable networking:
  Storage area networks with network coding,'' \emph{Allerton Conf.\ on
  Commun.\ Control and Comput.}, pp. 517--524, Oct. 2012.

\bibitem{berk_isit}
{L. Huang, S. Pawar, H. Zhang, and Kannan Ramchandran}, ``Codes can reduce
  queuing delay in data centers,'' \emph{Proc.\ Int.\ Symp.\ Inform.\ Theory},
  pp. 2766--2770, Jul. 2012.

\bibitem{mds_queue}
{N. Shah, K. Lee, and K. Ramachandran}, ``{The MDS queue: analyzing latency
  performance of codes and redundant requests},'' Tech. Rep. arXiv:1211.5405,
  Nov. 2012.

\bibitem{kabatiansky_krouk_semenov}
{G. Kabatiansky, E. Krouk and S. Semenov}, \emph{Error correcting coding and
  security for data networks:~analysis of the superchannel concept},
  1st~ed.\hskip 1em plus 0.5em minus 0.4em\relax Wiley, Mar. 2005, ch.~7.

\bibitem{Maxemchuk93}
N.~F. Maxemchuk, ``Dispersity routing in high-speed networks,'' \emph{Compu.\
  Networks and ISDN Sys.}, vol.~25, pp. 645--661, Jan. 1993.

\bibitem{mia}
Y.~Liu, J.~Yang, and S.~C. Draper, ``Exploiting route diversity in multi-packet
  transmission using mutual information accumulation,'' \emph{Allerton Conf.\
  on Commun.\ Control and Comput.}, pp. 1793--1800, Sep. 2011.

\bibitem{lihao_thesis}
L.~Xu, ``Highly available distributed storage systems,'' Ph.D. dissertation,
  California Institute of Technology, 1998.

\bibitem{Allerton10e}
E.~Soljanin, ``Reducing delay with coding in (mobile) multi-agent information
  transfer,'' \emph{Allerton Conf.\ on Commun.\ Control and Comput.}, pp.
  1428--1433, Sep. 2010.

\bibitem{mor_book}
M.~Harchol-Balter, \emph{Performance Modeling and Design of Computer Systems:
  Queueing Theory in Action}.\hskip 1em plus 0.5em minus 0.4em\relax Cambridge
  University Press, 2013.

\bibitem{kim_agarwal}
C.~Kim and A.~K. Agrawala, ``Analysis of the fork-join queue,'' \emph{IEEE
  Trans.\ Comput.}, vol.~38, no.~2, pp. 250--255, Feb. 1989.

\bibitem{nelson_tantawi}
R.~Nelson and A.~Tantawi, ``Approximate analysis of fork/join synchronization
  in parallel queues,'' \emph{IEEE Trans.\ Comput.}, vol.~37, no.~6, pp.
  739--743, Jun. 1988.

\bibitem{varki_merc_chen}
{E. Varki, A. Merchant and H. Chen}, ``The {M/M/1} fork-join queue with
  variable sub-tasks,'' {\it unpublished, available online.}

\bibitem{order_stat}
S.~Ross, \emph{A first course in probability}, 6th~ed.\hskip 1em plus 0.5em
  minus 0.4em\relax Prentice Hall, 2002, ch. 6.6, p. 273.

\bibitem{powerof2}
M.~Mitzenmacher, ``{The power of two choices in randomized load balancing},''
  Ph.D. dissertation, University of California Berkeley, CA, 1996.

\bibitem{flatto1984two}
L.~Flatto and S.~Hahn, ``Two parallel queues created by arrivals with two
  demands i,'' \emph{SIAM Journal on Applied Mathematics}, vol.~44, no.~5, pp.
  1041--1053, 1984.

\bibitem{wright1992two}
P.~E. Wright, ``Two parallel processors with coupled inputs,'' \emph{Advances
  in applied probability}, pp. 986--1007, 1992.

\bibitem{fj_stability}
{P. Konstantopoulos and J. Walrand}, ``Stationarity and stability of fork-join
  networks,'' \emph{{J.\ Appl.\ Prob.}}, vol.~26, pp. 604--614, Sep. 1989.

\bibitem{dsp_gallager}
H.~C. Tijms, \emph{A first course in stochastic models}, 2nd~ed.\hskip 1em plus
  0.5em minus 0.4em\relax Wiley, 2003, ch. 2.5, p.~58.

\bibitem{assoc_rand_vars}
{J. Esary, F. Proschan and D. Walkup}, ``Association of random variables, with
  applications,'' \emph{Annals of Math.\ Stat.}, vol.~38, no.~5, pp.
  1466--1474, Oct. 1967.

\bibitem{ArnoldGroeneveld}
B.~C. Arnold and R.~A. Groeneveld, ``Bounds on expectations of linear
  systematic statistics based on dependent samples,'' \emph{Annals of Stat.},
  vol.~7, pp. 220--223, Oct. 1979.

\bibitem{Papadatos}
N.~Papadatos, ``Maximum variance of order statistics,'' \emph{Ann.\ Inst.\
  Statist.\ Math}, vol.~47, pp. 185--193, 1995.

\bibitem{CrovellaBestavros}
M.~E. Crovella and A.~Bestavros, ``Self-similarity in {World Wide Web} traffic:
  evidence and possible causes,'' \emph{IEEE/ACM Trans.\ Networking}, pp.
  835--846, Dec. 1997.

\bibitem{FaloutsosFaloutsos}
M.~Faloutsos, P.~Faloutsos, and C.~Faloutsos, ``On power-law relationships of
  the {Internet} topology,'' \emph{Proc.\ ACM SIGCOMM}, pp. 251--262, 1999.

\bibitem{MapReduce}
J.~Dean and S.~Ghemawat, ``{MapReduce:} simplified data processing on large
  clusters,'' \emph{Commu.\ of ACM}, vol.~51, pp. 107--113, Jan. 2008.

\bibitem{VulimiriMichel}
A.~Vulimiri, O.~Michel, P.~B. Godfrey, and S.~Shenker, ``More is less: reducing
  latency via redundancy,'' \emph{Proc.\ ACM HotNets}, pp. 13--18, 2012.

\bibitem{BostoenMullender}
T.~Bostoen, S.~Mullender, and Y.~Berbers, ``Power-reduction techniques for
  data-center storage systems,'' \emph{ACM Comput.\ Surveys}, vol.~45, no.~3,
  2013.

\end{thebibliography}

\end{document}